\documentclass[runningheads]{llncs}
\usepackage[T1]{fontenc}
\usepackage{graphicx}
\usepackage{color}
\usepackage{graphicx} 

\usepackage{algorithm}
\usepackage[noend]{algorithmic}
\usepackage{subcaption}
\usepackage{soul}

\usepackage{tikz}
\usetikzlibrary{matrix,calc}
\usetikzlibrary{positioning}
\usetikzlibrary{shapes.misc}
\usetikzlibrary{calc}

\usepackage{pifont}
\newcommand{\cmark}{\ding{51}}%
\newcommand{\xmark}{\ding{55}}%

\usepackage{newfloat}
\usepackage{listings}
\floatstyle{ruled}
\newfloat{listing}{tb}{lst}{}
\floatname{listing}{Listing}

\usepackage{amsmath}
\usepackage{amsfonts}
\usepackage{mathtools}
\usepackage{amsmath,amssymb}
\newtheorem{assumption}{Assumption}
\DeclareMathOperator*{\argmax}{argmax}

\DeclareMathOperator*{\Argmax}{Argmax}
\DeclareMathOperator*{\Argmin}{Argmin}

\DeclareSymbolFont{extraup}{U}{zavm}{m}{n}
\DeclareMathSymbol{\varheart}{\mathalpha}{extraup}{86}
\DeclareMathSymbol{\vardiamond}{\mathalpha}{extraup}{87}

\newcommand{\shortminus}{\mkern0.1mu{-}\mkern0.1mu} 

\newcommand{\todo}[1]{\textcolor{red}{TODO: #1}}

\let\svthefootnote\thefootnote
\newcommand\blankfootnote[1]{%
  \let\thefootnote\relax\footnotetext{#1}%
  \let\thefootnote\svthefootnote%
}

\newif\ifExtendedVersion
\ExtendedVersiontrue
\begin{document}
\title{Multi-defender Security Games with Schedules}
%
%
\author{Zimeng Song \inst{1}
\and
Chun Kai Ling \inst{2} 
\and
Fei Fang\inst{2}}
%
%
\institute{Independent Researcher,
\email{zmsongzm@gmail.com} \and
Carnegie Mellon University,
\email{chunkail@cs.cmu.edu, feif@cs.cmu.edu}}
%
\maketitle              
\begin{abstract}
Stackelberg Security Games are often used to model strategic interactions in high-stakes security settings. The majority of existing models focus on single-defender settings where a single entity assumes command of all security assets. However, many realistic scenarios feature \textit{multiple heterogeneous defenders} with their own interests and priorities embedded in a more complex system. Furthermore, defenders rarely choose targets to protect. Instead, they have a multitude of defensive resources or \textit{schedules} at its disposal, each with different protective capabilities.
In this paper, we study security games featuring multiple defenders and schedules simultaneously. We show that unlike prior work on multi-defender security games, the introduction of schedules can cause non-existence of equilibrium even under rather restricted environments. We prove that under the mild restriction that any subset of a schedule is also a schedule, non-existence of equilibrium is not only avoided, but can be computed in polynomial time in games with two defenders. Under additional assumptions, our algorithm can be extended to games with more than two defenders and its computation scaled up in special classes of games with compactly represented schedules such as those used in patrolling applications.
Experimental results suggest that our methods scale gracefully with game size, making our algorithms amongst the few that can tackle multiple heterogeneous defenders.


\end{abstract}

\section{Introduction}
\blankfootnote{Equal contribution between authors Song and Ling.}
The past decades have seen a wave of interest in Stackelberg Security Games (SSG), with applications to infrastructure \cite{pita2009using}, wildlife poaching \cite{fang2017paws,fang2016deploying,fang2015security} and cybersecurity \cite{zarreh2018game,pawlick2019game,attiah2018game}.
SSGs are played between a single defender allocating defensive resources over a finite set of targets and an attacker who, after observing this allocation, best responds by attacking the target least defended \cite{tambe2011security}.  

Numerous variants of SSGs better reflecting the real world have been proposed.
Amongst the most well-researched extension are settings with \textit{scheduling constraints}. Instead of guarding a single target, the defender chooses a \textit{set} of targets instead, making it possible to model real world problems such as planning patrols for anti-poaching \cite{fang2017paws,fang2015security,wang2019deep,basak2016combining} and the US coast guard \cite{shieh2012protect,pita2008armor}, infrastructure protection \cite{jain2010security}, and optimal placement of police checkpoints \cite{jain2013security}.
Another less explored extension is the \textit{multi-defender setting}, where multiple defenders (e.g., city, state and federal law enforcement), each utilize distinct resources in tandem, potentially resulting in miscoordination \cite{jiang2013defender}. 
Recent work in \cite{gan2018stackelberg,gan2022defense} solved the challenging problem of finding equilibrium when defenders are heterogeneous, both in the coordinated and uncoordinated case.

Unfortunately, there is almost no literature on settings exhibiting both scheduling \textit{and} multiple defenders. 
Our work fills this non-trivial gap.
In contrast to the positive results of \cite{gan2018stackelberg,gan2022defense}, we show that equilibrium may not exist with the inclusion of scheduling constraints, even under extremely stringent constraints on other aspects of the problem. 
We then guarantee existence of equilibrium in two defender settings when restricted to schedules satisfying the \textit{subset-of-a-schedule-is-also-a-schedule} structure \cite{korzhyk2010complexity}, on top of other mild restrictions. 
We construct polynomial time algorithms for computing equilibrium in such restricted settings and propose two extensions. The first utilizes an additional assumption of Monotone Schedules to handle the general multi-defender setting.
The second scales to scenarios with a large (possibly exponential) number of schedules in structured domains such as patrolling. 
Empirically, our algorithms scale gracefully, making them viable for use in the real world.

\section{Background and Related Work}
\label{sec:background}
Our work is motivated by the security application played on the layered network in Figure~\ref{fig:network-path}, where vertices represent neighborhoods and edges represent connecting roads.
Distinct law enforcement agencies (e.g., local and federal police) patrol along a path starting from vertex $a$ to vertex $e$. Patrols provide defence, or \textit{coverage} to the neighborhoods they pass (Figure~\ref{fig:network-path}). 
By randomizing or splitting their patrols, agencies can broaden coverage at the expense of thinning them (Figure~\ref{fig:network-split}). Coverage at each neighborhood is accumulated over patrols (Figure~\ref{fig:network-combined}).
An attacker chooses a single neighborhood to attack based on this coverage, giving negative reward to law enforcement agencies. Neighborhoods and law enforcement are non-homogeneous: neighborhoods differ in density and demographics, local police value local businesses and inhabitants, while federal agencies focus on federal government assets.  
Given these competing objectives, how should law enforcement agencies plan their patrols?
\begin{figure}[t]
\centering
\begin{subfigure}[t]{0.32\textwidth}
    \centering
\begin{tikzpicture}[>=stealth, every node/.style={circle, draw, minimum size=12pt, inner sep=0pt}]
  \definecolor{C4}{RGB}{0, 0, 0}
\definecolor{C3}{RGB}{60, 60, 60}
\definecolor{C2}{RGB}{120, 120, 120}
\definecolor{C1}{RGB}{180, 180, 180}
\definecolor{C0}{RGB}{240, 240, 240}
  
  \node [fill=C2](source) at (0,0.33)     {$a$};
  \node (m1n1) at (0.8,-.66)        {};
  \node (m1n2) at (0.8,0)         {};
  \node [fill=C2](m1n3) at (0.8,.66)         {$b_2$};
  \node (m1n4) at (0.8,1.33)         {};
  \node (m2n1) at (1.6,-.66)        {};
  \node [fill=C2](m2n2) at (1.6,0)         {$c_3$};
  \node (m2n3) at (1.6,.66)         {};
  \node (m2n4) at (1.6,1.33)         {};
  \node [fill=C2](m3n1) at (2.4,-.66)        {$d_4$};
  \node (m3n2) at (2.4,0)         {};
  \node (m3n3) at (2.4,.66)         {};
  \node (m3n4) at (2.4,1.33)         {};
  \node [fill=C2](sink) at (3.2,0.33)       {$e$};

  \draw[->] (source) -- (m1n1);
  \draw[->] (source) -- (m1n2);
  \draw[->,line width=1mm, green] (source) -- (m1n3);
  \draw[->] (source) -- (m1n4);

  \draw[->,line width=1mm, green] (m3n1) -- (sink);
  \draw[->] (m3n2) -- (sink);
  \draw[->] (m3n3) -- (sink);
  \draw[->] (m3n4) -- (sink);
  
  \draw[->] (m1n1) -- (m2n1);
  \draw[->] (m1n1) -- (m2n2);
  
  \draw[->] (m1n2) -- (m2n1);
  \draw[->] (m1n2) -- (m2n2);
  \draw[->] (m1n2) -- (m2n3);
  
  \draw[->,line width=1mm, green] (m1n3) -- (m2n2);
  \draw[->] (m1n3) -- (m2n3);
  \draw[->] (m1n3) -- (m2n4);
  
  \draw[->] (m1n4) -- (m2n3);
  \draw[->] (m1n4) -- (m2n4);

  \draw[->] (m2n1) -- (m3n1);
  \draw[->] (m2n1) -- (m3n2);
  
  \draw[->,line width=1mm, green] (m2n2) -- (m3n1);
  \draw[->] (m2n2) -- (m3n2);
  \draw[->] (m2n2) -- (m3n3);
  
  \draw[->] (m2n3) -- (m3n2);
  \draw[->] (m2n3) -- (m3n3);
  \draw[->] (m2n3) -- (m3n4);
  
  \draw[->] (m2n4) -- (m3n3);
  \draw[->] (m2n4) -- (m3n4);
\end{tikzpicture}
\caption{}
\label{fig:network-path}
\end{subfigure}
\begin{subfigure}[t]{0.32\textwidth}
    \centering
\begin{tikzpicture}[>=stealth, every node/.style={circle, draw, minimum size=12pt, inner sep=0pt}]
  \definecolor{C4}{RGB}{0, 0, 0}
\definecolor{C3}{RGB}{60, 60, 60}
\definecolor{C2}{RGB}{120, 120, 120}
\definecolor{C1}{RGB}{180, 180, 180}
\definecolor{C0}{RGB}{240, 240, 240}
  
  \node [fill=C2](source) at (0,0.33)     {a};
  \node [fill=C1](m1n1) at (0.8,-.66)        {$b_4$};
  \node (m1n2) at (0.8,0)         {};
  \node [fill=C1](m1n3) at (0.8,.66)         {$b_2$};
  \node (m1n4) at (0.8,1.33)         {};
  \node (m2n1) at (1.6,-.66)        {};
  \node [fill=C2](m2n2) at (1.6,0)         {$c_3$};
  \node (m2n3) at (1.6,.66)         {};
  \node (m2n4) at (1.6,1.33)         {};
  \node [fill=C1](m3n1) at (2.4,-.66)        {$d_4$};
  \node (m3n2) at (2.4,0)         {};
  \node [fill=C1] (m3n3) at (2.4,.66)         {$d_2$};
  \node (m3n4) at (2.4,1.33)         {};
  \node [fill=C2](sink) at (3.2,0.33)       {e};

  \draw[->,line width=.5mm,green] (source) -- (m1n1);
  \draw[->] (source) -- (m1n2);
  \draw[->,line width=.5mm,green] (source) -- (m1n3);
  \draw[->] (source) -- (m1n4);

  \draw[->,line width=.5mm,green] (m3n1) -- (sink);
  \draw[->] (m3n2) -- (sink);
  \draw[->,line width=.5mm,green] (m3n3) -- (sink);
  \draw[->] (m3n4) -- (sink);
  
  \draw[->] (m1n1) -- (m2n1);
  \draw[->,line width=.5mm,green] (m1n1) -- (m2n2);
  
  \draw[->] (m1n2) -- (m2n1);
  \draw[->] (m1n2) -- (m2n2);
  \draw[->] (m1n2) -- (m2n3);
  
  \draw[->,line width=.5mm,green] (m1n3) -- (m2n2);
  \draw[->] (m1n3) -- (m2n3);
  \draw[->] (m1n3) -- (m2n4);
  
  \draw[->] (m1n4) -- (m2n3);
  \draw[->] (m1n4) -- (m2n4);

  \draw[->] (m2n1) -- (m3n1);
  \draw[->] (m2n1) -- (m3n2);
  
  \draw[->,line width=.5mm,green] (m2n2) -- (m3n1);
  \draw[->] (m2n2) -- (m3n2);
  \draw[->,line width=.5mm,green] (m2n2) -- (m3n3);
  
  \draw[->] (m2n3) -- (m3n2);
  \draw[->] (m2n3) -- (m3n3);
  \draw[->] (m2n3) -- (m3n4);
  
  \draw[->] (m2n4) -- (m3n3);
  \draw[->] (m2n4) -- (m3n4);  
\end{tikzpicture}
\caption{}
\label{fig:network-split}
\end{subfigure}
\begin{subfigure}[t]{0.32\textwidth}
    \centering
\begin{tikzpicture}[>=stealth, every node/.style={circle, draw, minimum size=12pt, inner sep=0pt}]
  \definecolor{C4}{RGB}{0, 0, 0}
\definecolor{C3}{RGB}{60, 60, 60}
\definecolor{C25}{RGB}{90, 90, 90}
\definecolor{C2}{RGB}{120, 120, 120}
\definecolor{C15}{RGB}{150, 150, 150}
\definecolor{C1}{RGB}{180, 180, 180}
\definecolor{C05}{RGB}{210, 210, 210}
\definecolor{C0}{RGB}{240, 240, 240}
  
  \node [fill=C4](source) at (0,0.33)     {a};
  \node [fill=C1](m1n1) at (0.8,-.66)        {};
  \node [fill=C1](m1n2) at (0.8,0)         {};
  \node [fill=C1](m1n3) at (0.8,.66)         {};
  \node [fill=C1](m1n4) at (0.8,1.33)         {};
  \node [fill=C05] (m2n1) at (1.6,-.66)        {};
  \node [fill=C25](m2n2) at (1.6,0)         {};
  \node [fill=C05](m2n3) at (1.6,.66)         {};
  \node [fill=C05](m2n4) at (1.6,1.33)         {};
  \node [fill=C15](m3n1) at (2.4,-.66)        {};
  \node [fill=C05](m3n2) at (2.4,0)         {};
  \node [fill=C1] (m3n3) at (2.4,.66)         {};
  \node [fill=C1](m3n4) at (2.4,1.33)         {};
  \node [fill=C4](sink) at (3.2,0.33)       {e};

  \draw[->,line width=.5mm,green] (source) -- (m1n1);
  \draw[->,line width=0.5mm,blue] (source) -- (m1n2);
  \draw[->,line width=.5mm,green] (source) -- (m1n3);
  \draw[->,line width=0.5mm,blue] (source) -- (m1n4);

  \draw[->,line width=.75mm,orange] (m3n1) -- (sink);
  \draw[->,line width=0.25mm,blue] (m3n2) -- (sink);
  \draw[->,line width=.5mm,green] (m3n3) -- (sink);
  \draw[->,line width=0.5mm,blue] (m3n4) -- (sink);
  
  \draw[->] (m1n1) -- (m2n1);
  \draw[->,line width=.5mm,green] (m1n1) -- (m2n2);
  
  \draw[->,line width=0.25mm,blue] (m1n2) -- (m2n1);
  \draw[->,line width=0.25mm,blue] (m1n2) -- (m2n2);
  \draw[->] (m1n2) -- (m2n3);
  
  \draw[->,line width=.5mm,green] (m1n3) -- (m2n2);
  \draw[->] (m1n3) -- (m2n3);
  \draw[->] (m1n3) -- (m2n4);
  
  \draw[->,line width=0.25mm,blue] (m1n4) -- (m2n3);
  \draw[->,line width=0.25mm,blue] (m1n4) -- (m2n4);

  \draw[->,line width=0.25mm,blue] (m2n1) -- (m3n1);
  \draw[->] (m2n1) -- (m3n2);
  
  \draw[->,line width=.5mm, green] (m2n2) -- (m3n1);
  \draw[->,line width=0.25mm,blue] (m2n2) -- (m3n2);
  \draw[->,line width=.5mm, green] (m2n2) -- (m3n3);
  
  \draw[->] (m2n3) -- (m3n2);
  \draw[->] (m2n3) -- (m3n3);
  \draw[->,line width=0.25mm,blue] (m2n3) -- (m3n4);
  
  \draw[->] (m2n4) -- (m3n3);
  \draw[->,line width=0.25mm,blue] (m2n4) -- (m3n4);  
\end{tikzpicture}
\caption{}
\label{fig:network-combined}
\end{subfigure}
\caption{Multi-defender Patrolling on Layered Networks. Vertices are targets and black edges are roads connecting targets. (a) The green path shows a patrol route by a single defender, passing through $b_2$, $c_3$ and $d_4$. Protected nodes are shaded, with the darkness indicating degree of protection.
(b) A coverage obtained by splitting/randomizing patrols. Thinner arrows indicate smaller patrols. Compared to Figure~\ref{fig:network-path}, more targets are covered, but with lower intensity.
(c) A joint patrol with two defenders; the first employs the route in Figure~\ref{fig:network-split} and the second follows the blue paths. The orange edge is used by both defenders.
}
\label{fig:network-all}
\end{figure}
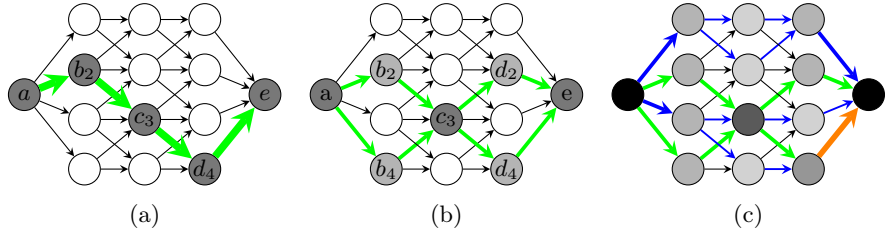

The closest pieces of work to us are by Lou et. al. \cite{lou2015equilibrium} and Gan et. al. \cite{gan2018stackelberg,gan2022defense}. The work by Gan et. al. \cite{gan2022defense} focuses on the case with multiple heterogeneous defenders without schedules
, showing that an equilibrium always exists and can be computed in polynomial time in both the coordinated and uncoordinated case by extending the classic water-filling algorithm \cite{kiekintveld2009computing}. 
However, their model is limited by the lack of schedules. This is not merely an issue of computation: our work shows that the inclusion of schedules can trigger non-existence of equilibrium.
The work by Lou et. al. \cite{lou2015equilibrium} consider multiple defenders and scheduling, analyzing equilibrium in terms of the Price of Anarchy between defenders and giving conditions for their existence.\footnote{We point out \cite{lou2015equilibrium} and \cite{gan2022defense} use different tie-breaking models for the attacker.} However, the bulk of their work concerns homogeneous defenders; 
their results for heterogeneous defenders limited (e.g., having every target completely covered). Here, they acknowledge possible non-existence of equilibrium, proposing approximate solvers based on Mixed-Integer Program with quadratic constraints. 
In contrast, our work makes additional assumptions 
but guarantees existence and polynomial time solutions. 

SSGs as a whole have a long and illustrious history. First introduced by von Stackelberg \cite{von1934marktform} to model competitive firms and the first mover’s advantage, it saw a resurgence beginning with \cite{von2010leadership} alongside a wave of applications primarily in the domain of security which modeled defenders as first movers or leaders \cite{tambe2011security,an2017stackelberg}. Since then, an enormous amount of literature has surfaced, e.g., computing equilibrium in sequential settings \cite{letchford2010computing,vcerny2018incremental,ling2021safe}, handling bounded rationality of defenders \cite{kar2017trends}, and various other structural assumptions such games on networks \cite{letchford2012networkcomputing}, each catering to different variants of security applications.

\section{Nash-Stackelberg Equilibrium with Scheduling}


Our setting involves $n$ heterogeneous defenders, $T$ heterogeneous targets and a single attacker.
Each defender allocates defensive resources which induce \textit{coverage} over targets, a quantitative measure of the degree to which each target is protected.
For example, coverage can refer to the average number of police officers patrolling at a particular neighborhood (Figure~\ref{fig:network-all}).
As is customary in security games, we employ \textit{Stackelberg leadership} models.
Each defender first independently \textit{commits} to its coverage. The attacker then chooses to attack a target $t \in [T]$ with the lowest total coverage under this commitment.


Formally, each defender $i$ has an \textit{attainable set} of coverage $V^i \subset \mathbb{R}_+^T$ (which we define concretely later), 
from which it chooses a \textit{coverage vector} $v^i \in V^i$, where
defender $i$ contributes $v^i(j)$ coverage to target $j$.
A \textit{coverage profile} $\mathbf{v} = (v^1, v^2, \cdots ,v^n) \in V^1 \times V^2 \cdots \times V^n$ is an ordered tuple of coverages from each defender.
We assume that coverage accumulates across defenders additively, such that for a given coverage profile $\mathbf{v}$, the \textit{total coverage} of all defenders is $v^{total} = \sum_{i\in [n]} v^i$. 
The total coverage on target $j$ is $v^{total}(j) = \sum_{i\in [n]} v^i(j) \geq 0$. 





We assume \textit{coverage independent payoffs}: each defender's payoff is based on the attacked target, but \textit{not} on the coverage itself. In practical terms, this means that an attack would ``always succeed'', and the purpose of coverage is to redirect the attacker elsewhere. For example, security officers may not be able to prevent a determined terrorist attack, but having more officers in crowded areas may make the cost of attack prohibitively high such that the attack occurs at a less populated area.
This assumption means there is no need to work explicitly with numerical values for defender payoffs. 
Instead, each defender $i$ has a fixed order of preference over targets given by the total order $\succ_i$. We write $j \succ_i k$ (resp. $j \prec_i k$ ) if defender $i$ prefers target $j$ over $k$ to be attacked (resp. not attacked). 
We assume that there are no ties in defender preferences, hence $j=_i k$ if and only if $j=k$. 
We write $j \succeq_i k$ if and only if $j \succ_i k$ or $j =_i k$, with $j \preceq_i k$ being defined analogously. 
Preference orders differ between defenders, hence $j \succ_i k$ does not imply $j \succ_{i'} k$ for $i \neq i'$. 
For any target $t \in [T]$, we define the set $\mathcal{T}^{\succ_i}_t = \left\{k \in [T] | k \succ_i t \right\}$, i.e., the set of targets which defender $i$ strictly prefers to be attacked over target $t$. $\mathcal{T}^{\prec_i}_t$, $\mathcal{T}^{\succeq_i}_t$, and $\mathcal{T}^{\preceq_i}_t$ are defined analogously. 

We now define formally define the Nash-Stackelberg equilibrium (NSE) given $V^i$ and $\succ_i$ for each defender. We call a tuple $(v^1, \dots, v^n, t) \in V^1 \times \dots \times V^n \times [T]$ a strategy profile, abbreviated by $(\mathbf{v}, t)$.
Strategy profiles are an NSE when (i) the attacker is attacking a least covered target and (ii) neither defender $i$ is willing to unilaterally deviate their coverage from $v^i$ to $\widehat{v}^i \in V^i$ (written as $v^i \rightarrow \widehat{v}^i$), assuming that the attacker could react to this deviation by possibly adjusting its target from $t$ to $\widehat{t}$. 
The NSE is named as such because the attacker best responds to the total coverage as if it was a Stackelberg follower, while defenders interact amongst themselves as if it were Nash.
The former condition is formalized easily.
\begin{definition}[Attacker's Best Response Set] Given a total coverage $v^{total} \in \mathbb{R}_+^T$, its attacker best response set is
\footnote{In this paper, we use superscripts $\left( \cdot \right) ^ i$ to specify a defender's index, 
and brackets $\left( \right)$ for elements in a vector. We capitalize $\Argmax$ to denote the subset of maximal elements, and lower case $\argmax$ when referring to an arbitrary one.  
}
\begin{align}
B(v^{total}) = \Argmin_{t \in [T]} v^{total}(t)
= \left\{ t \in [T] \Big| v^{total}(t) = \min_{t' \in [T]} v^{total}(t') \right\}
\end{align}
\end{definition}
\begin{definition}[Attacker Incentive Compatibility (AIC)] A strategy profile $(\mathbf{v}, t)$ is \textit{attacker incentive compatible} if and only if $t \in B(v^{total})$.
\end{definition}
Essentially, a tuple $(\mathbf{v}, t)$ is AIC when the attacker does not strictly prefer some target $\widehat{t} \neq t$ over $t$. Naturally, we desire profiles which are AIC.

As noted by \cite{gan2018stackelberg}, condition (ii) contains an important nuance: if defender $i$ deviates via $v^i \rightarrow \widehat{v}^i$ such that $\widehat{v}^{total}$ does not have a unique minimum coverage (i.e., $|B(\widehat{v}^{total})| > 1$), how should the attacker break ties? 
In single-defender scenarios, one typically breaks ties in favor of the defender, also called the \textit{Strong} Stackelberg equilibrium.
With multiple defenders, this is ambiguous since defenders are heterogeneous.
This is a non-trivial discussion, since the way the attacker breaks ties significantly affects the space of equilibrium.
There are two natural choices for breaking ties: either punish or prioritize the deviating defender. These choices mirror the concepts of weak and strong Stackelberg equilibrium in single-defender settings. 
As \cite{gan2018stackelberg} argue, the second convention of benefiting the deviating defender can lead to spurious deviations from defenders. For example, a defender may make a trivial ``identity'' deviation from $v^i \rightarrow v^i$, yet demand a nontrivial change in attacked target.
As such we adopt the former concept. This means that deviating players are pessimistic towards any change in attacked target after their deviation (this pessimism exists only for deviations).


\begin{definition}[Defender $i$-Weakly Attacker Incentive Compatibility ($i$-WAIC)] A strategy profile $(\mathbf{v}, t)$ is $i$-WAIC if and only if (i) $t \in B(v^{total})$ and (ii) $\widehat{t} \in B(v^{total}) \implies \widehat{t} \succeq_i t$. 
\end{definition}

\begin{definition}[Defender $i$-Incentive Compatibility ($i$-IC)] A strategy profile $(\mathbf{v}, t)$ is $i$-IC if and only if there does not exist $\widehat{v}^i\in V^i$ and $\widehat{t} \succ_i t$ such that $(v^1, v^2, \cdots, v^{i - 1}, \widehat{v}^i, v^{i+1}, \cdots, v^n, \widehat{t})$ is $i$-WAIC.
\end{definition}

\begin{definition}[Nash-Stackelberg Equilibrium (NSE)] A strategy profile $(\mathbf{v}, t)$ is an NSE if and only if it is AIC and $i$-IC for all $i \in [n]$.    
\label{def:nse}
\end{definition}
    Put simply, $(\mathbf{v}, t)$ is $i$-WAIC if it is AIC and the choice of $t$ is made such as to break ties \textit{against} defender $i$ (clearly, this is a strictly stronger condition than AIC). 
Consequently, $(\mathbf{v}, t)$ is $i$-IC if there is no deviation $v^i \rightarrow \widehat{v}^i$ such that defender $i$ benefits strictly when the attacker changes it's best-response $t \rightarrow \widehat{t}$ while breaking ties against defender $i$.
    Thus, any candidate NSE $(\mathbf{v}, t)$ is only required be AIC, allowing us to freely choose how the attacker tiebreaks, as though $t$ was the ``agreed upon norm'' target for the attacker. 
    However, if defender $i$ deviates, we use the stronger notion of WAIC for post-deviation tiebreaking.
\begin{remark}
    \label{rem:differences-prior}
    Our model does not completely generalize \cite{gan2018stackelberg} due to coverage independent payoffs and additive coverage. However, these assumptions are related to their cases of correlated defenders and non-overlapping payoffs respectively.
\end{remark}

\section{Analysis and Algorithms} 
Clearly, the existence and computation of NSE depends on the set $V^i$.
The simplest specification of $V^i$ is obtained by explicitly specifying schedules and requiring \textit{clearance constraints}.
Suppose each defender $i \in [n]$ has $1$ unit of divisible defensive resource to allocate across a set of $S^i > 0$ \textit{schedules} $\mathcal{S}^i = \{ s^i_1, s^i_2,\cdots s^i_{S^i} \}$. 
Each $s^i_z \in \mathbb{R}^T_+$ is the non-negative coverage over all $T$ targets when defender $i$ allocates its entire defensive resource to the $z$-th schedule. 
$s^i_z(j), j \in [T]$ is the coverage specific to target $j$. Each defender $d_i$'s \textit{strategy} is a distribution $x^i \in \mathbb{R}^{S^i}_+$ over its set of schedules, where $\sum_{s_z^i \in \mathcal{S}^i} x^i(s_z^i) = 1$. 


\begin{assumption}[Clearance]
Given schedules $\mathcal{S}^i$, the coverage vector $v^i$ is said to satisfy clearance if all defensive resources are fully utilized, i.e., 
\begin{align*}
V^i = \left \{ 
v \in \mathbb{R}_+^T \Bigg| \exists x^i \ \text{such that } \forall j \in [T] \ v(j) = \sum_{s_z^i \in \mathcal{S}^i} x^i(s_z^i) \cdot s_{z}^i (j)
\right \}.
\end{align*}
\label{ass:clearance-coverage}
\end{assumption}
The clearance constraint on $V^i$ essentially requires defenders to expend as much as they can, preventing them from ``slacking off''. 
When each schedule attacks a distinct target, i.e, $S^i = \{ e_1, e_2, \dots, e_T \} $, clearance constraints reduce to the setting of \cite{gan2018stackelberg} and equilibria will exist. 
However, our focus is on the scheduled setting. 
Indeed, we now demonstrate an instance where NSE do not exist.\footnote{An earlier version of this paper included an incorrect example.}
\begin{figure}[ht]
    \centering
    \begin{subfigure}[t]{0.15\linewidth}
    \centering
    \begin{tikzpicture}
    \centering 
    \matrix (x) [matrix of math nodes,inner xsep=-4em,
    column sep=1em, row sep =1em,
    minimum width=1.3em,
    text height=.6em,
    text depth=-0em,
    anchor=center] {%
    \textcolor{green}{11} & \textcolor{blue}{12} \\
    \textcolor{blue}{21} & \textcolor{green}{22} \\};
    \foreach \i in {1,2} 
        \draw ($(x-1-\i.north west)+(-.5em,.5em)$) -- ($(x-2-\i.south west)+(-.5em,-.5em)$);
    \foreach \i in {1,2} 
        \draw ($(x-\i-1.south west)+(-.5em,-.5em)$) -- ($(x-\i-2.south east)+(.5em,-.5em)$);
    \draw ($(x-1-1.north west)+(-.5em,.5em)$) -| ($(x-2-2.south east)+(.5em,-.5em)$);
    \end{tikzpicture}
    \caption{}
    \label{fig:counterexample-simple}
    \end{subfigure}
    \begin{subfigure}[t]{0.15 \linewidth}
    \centering 
    \begin{tikzpicture}
    \matrix (x) [matrix of math nodes,inner xsep=-4em,
    column sep=1em, row sep =1em,
    minimum width=1.3em,
    text height=.6em,
    text depth=-.0em,
    anchor=center] {%
    1 \shortminus \epsilon & 1 \\
    k \cdot \epsilon & 0 \\};
    \foreach \i in {1,2} 
        \draw ($(x-1-\i.north west)+(-.5em,.5em)$) -- ($(x-2-\i.south west)+(-.5em,-.5em)$);
    \foreach \i in {1,2} 
        \draw ($(x-\i-1.south west)+(-.5em,-.5em)$) -- ($(x-\i-2.south east)+(.5em,-.5em)$);
    \draw ($(x-1-1.north west)+(-.5em,.5em)$) -| ($(x-2-2.south east)+(.5em,-.5em)$);
    \end{tikzpicture}
    \caption{$s_1^1$}
    \label{fig:s11}
    \end{subfigure}
    \begin{subfigure}[t]{0.15 \linewidth}
    \centering 
    \begin{tikzpicture}
    \matrix (x) [matrix of math nodes,inner xsep=-4em,
    column sep=1em, row sep =1em,
    minimum width=1.3em,
    text height=.6em,
    text depth=-.0em,
    anchor=center] {%
    0 & k \cdot \epsilon \\
    1 & 1 \shortminus \epsilon \\};
    \foreach \i in {1,2} 
        \draw ($(x-1-\i.north west)+(-.5em,.5em)$) -- ($(x-2-\i.south west)+(-.5em,-.5em)$);
    \foreach \i in {1,2} 
        \draw ($(x-\i-1.south west)+(-.5em,-.5em)$) -- ($(x-\i-2.south east)+(.5em,-.5em)$);
    \draw ($(x-1-1.north west)+(-.5em,.5em)$) -| ($(x-2-2.south east)+(.5em,-.5em)$);
    \end{tikzpicture}
    \caption{$s_2^1$}
    \label{}
    \end{subfigure}
    \begin{subfigure}[t]{0.15 \linewidth}
    \centering
    \begin{tikzpicture}
    \matrix (x) [matrix of math nodes,inner xsep=-4em,
    column sep=1em, row sep =1em,
    minimum width=1.3em,
    text height=.6em,
    text depth=-.0em,
    anchor=center] {%
    1 & 0 \\
    1 \shortminus \epsilon & k \cdot \epsilon \\};
    \foreach \i in {1,2} 
        \draw ($(x-1-\i.north west)+(-.5em,.5em)$) -- ($(x-2-\i.south west)+(-.5em,-.5em)$);
    \foreach \i in {1,2} 
        \draw ($(x-\i-1.south west)+(-.5em,-.5em)$) -- ($(x-\i-2.south east)+(.5em,-.5em)$);
    \draw ($(x-1-1.north west)+(-.5em,.5em)$) -| ($(x-2-2.south east)+(.5em,-.5em)$);
    \end{tikzpicture}
    \caption{$s_1^2$}
    \label{}
    \end{subfigure}
    \begin{subfigure}[t]{0.15 \linewidth}
    \centering
    \begin{tikzpicture}
    \matrix (x) [matrix of math nodes,inner xsep=-4em,
    column sep=1em, row sep =1em,
    minimum width=1.3em,
    text height=.6em,
    text depth=-.0em,
    anchor=center] {%
    k \cdot\epsilon & 1 \shortminus \epsilon \\
    0 & 1 \\};
    \foreach \i in {1,2} 
        \draw ($(x-1-\i.north west)+(-.5em,.5em)$) -- ($(x-2-\i.south west)+(-.5em,-.5em)$);
    \foreach \i in {1,2} 
        \draw ($(x-\i-1.south west)+(-.5em,-.5em)$) -- ($(x-\i-2.south east)+(.5em,-.5em)$);
    \draw ($(x-1-1.north west)+(-.5em,.5em)$) -| ($(x-2-2.south east)+(.5em,-.5em)$);
    \end{tikzpicture}
    \caption{$s_2^2$}
    \label{}
    \end{subfigure}
    \begin{subfigure}[t]{0.15\linewidth}
    \centering
    \begin{tikzpicture}
    \centering 
    \matrix (x) [matrix of math nodes,inner xsep=-4em,
    column sep=1em, row sep =1em,
    minimum width=1.3em,
    text height=.6em,
    text depth=-.0em,
    anchor=center] {%
    11 & 12 \\
    21 & 22 \\};
    \draw[->,color=blue,line width=0.5mm] (x-2-2) -- (x-2-1);
    \draw[->,color=green,line width=0.5mm] (x-1-2) -- (x-2-2);
    \draw[->,color=blue,line width=0.5mm] (x-1-1) -- (x-1-2);
    \draw[->,color=green,line width=0.5mm] (x-2-1) -- (x-1-1);
    \foreach \i in {1,2} 
        \draw ($(x-1-\i.north west)+(-.5em,.5em)$) -- ($(x-2-\i.south west)+(-.5em,-.5em)$);
    \foreach \i in {1,2} 
        \draw ($(x-\i-1.south west)+(-.5em,-.5em)$) -- ($(x-\i-2.south east)+(.5em,-.5em)$);
    \draw ($(x-1-1.north west)+(-.5em,.5em)$) -| ($(x-2-2.south east)+(.5em,-.5em)$);
    \end{tikzpicture}
    \caption{}
    \label{fig:counterexample-all}
    \end{subfigure}
    \caption{Illustration of Example~\ref{example:counterexample}. (a) Target labels. Defender $1$ prefers diagonal (green) targets attacked; defender $2$ prefers off diagonal (blue) targets. 
    (b-e) Defender schedules $s_1^1, s_2^1, s_1^2$ and $s_2^2$. (f) Cyclic behavior. Green/blue arrows indicate changes in targets induced by defender $1$/$2$ deviating.}
    \label{fig:counterexample-existence}
\end{figure}
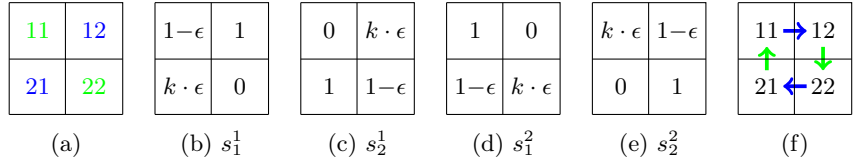
\vspace{-5mm}
\begin{example}
    The game in Figure~\ref{fig:counterexample-existence} has $n=2$ defenders and $T=4$ targets. Targets $11,12,21,22$ are organized in a $2\times 2$ matrix.
    Defender $1$ has preferences $22 \succ_1 11 \succ_1 12 \succ_1 21$ (i.e., prefers diagonal targets attacked) while defender $2$ has them in reverse, $21 \succ_2 12 \succ_2 11 \succ_2 22$ (i.e., prefers off diagonal targets attacked).
    Each has $2$ schedules, $s_1^1=(1-\epsilon, 1, k\epsilon, 0)$, $s_2^1=(0, k\epsilon, 1, 1-\epsilon)$ and $s_1^2 = (1, 0, 1-\epsilon, k\epsilon)$, $s_2^2=(k\epsilon, 1-\epsilon, 0, 1)$, where $k \geq 1$, $0 \leq \epsilon \ll 1$ and $k\epsilon < 1$.
    \label{example:counterexample}
\end{example}
Suppose (for now) that $\epsilon=0$ in Example~\ref{example:counterexample}. Then, defender $1$ decides how to split its coverage across rows, while defender $2$ the columns. 
Suppose $(\mathbf{v}, t)$ is a NSE. If $t = 11$, defender $2$ is incentivized to deviate to $\widehat{v}^2 = s_1^2$ regardless of $v^1$, since this always causes $12$ to have the ``lowest'' coverage and $12 \succ_2 11$. The same may be said for $t=22$, $t=12$ and $t=21$, where the latter two have defender $1$ deviating. This ``cyclic'' behavior (Figure~\ref{fig:counterexample-all}) implies that no equilibrium exists.

The above argument is only partially correct. When $\epsilon=0$, Example~\ref{example:counterexample} has a NSE $(\mathbf{v}, 11)$ where $v^1=v^2=(0.5, 0.5, 0.5, 0.5)$. If defender $2$ deviates to $\widehat{v}^2 = s_1^2$, 
the attacker tiebreaks, choosing target $22$ over target $12$, which hurts defender $2$. It may be verified that this is indeed a NSE. Introducing $\epsilon=10^{-3}$ and $k=100$ fixes this, ensuring target $12$ possesses the \textit{strictly} lowest coverage after deviation regardless of $v^1$. 
\ifExtendedVersion
The full derivation is deferred to the Appendix.
\else
The full derivation is deferred to the extended version. 
\fi

The non-existence is caused by the rigid enforcement of schedules. For example, even though defender $1$ prefers $11$ to be attacked over $12$, defending $11$ through schedule $s^1_1$ forces it to simultaneously defend $12$, which it rather not defend. 
Our fix expands the attainable coverage $V^i$: 
instead of clearance, defenders may provide \textit{less coverage} than what they could have. 
This assumption was used by \cite{korzhyk2010complexity} and is quite reasonable. For example, a patroller may deliberately let down their guard at areas of lower priority, \textit{encouraging} attacks there. 

\begin{assumption}[Subset-of-a-Schedule-is-Also-a-Schedule (SSAS)]
\begin{align*}
V^i = \left \{ 
v \in \mathbb{R}_+^T \Bigg| \exists x^i \text{ such that } \forall j \in [T] \ v(j) \leq \sum_{s_z^i \in \mathcal{S}^i} x^i(s_z^i) \cdot s_{z}^i (j)
\right \}.
\end{align*}
\label{ass:SSAS}
\end{assumption}

Assumption~\ref{ass:SSAS} is obtained from Assumption~\ref{ass:clearance-coverage} by replacing the equality constraint by an inequality. 
Checking if $v^i \in V^i$ is done efficiently using linear programs. 
Clearly, $V^i$ under SSAS is a superset of that under clearance. 
Under SSAS, Example~\ref{example:counterexample} with $\epsilon=0$ has a NSE $(v^1, v^2, 11)$, where $v^1=(0, 0.5, 0.5, 0)$, $v_2 = (0, 0, 0, 1.0)$. 
\ifExtendedVersion
Details, alongside the $\epsilon > 0$ case are in the Appendix.
\else
Details, alongside the $\epsilon > 0$ case are in the extended version.
\fi

\subsection{Existence and Computation of NSE Assuming SSAS}
For now, we restrict ourselves to $2$ defenders.
Under SSAS, 
any NSE can be converted into a simpler canonical form, which we exploit to guarantee existence for all $\mathcal{S}^i$, together with 
a polynomial time algorithm. 
We present these results via a series of reductions. \textbf{All proofs are deferred to the appendix.}
\begin{lemma} 
    If $(v^1, v^2, t)$ is an NSE, then there exists another NSE $(\widetilde{v}^1, v^2, t)$ 
    such that (i) $\widetilde{v}^1(t) = 0$ and (ii) $\widetilde{v}^1(j)=v^1(j)$ for all targets $j\neq t \in [T]$.
    \label{lem:step1}
\end{lemma} 

Applying Lemma~\ref{lem:step1} to each defender guarantees that each NSE $(\mathbf{v}, t)$ must correspond to an NSE with zero coverage on $t$. 
If we knew the attacked target $t$, we can reduce our search to coverage profiles with zero coverage on $t$.
\begin{figure}[t]
    \centering
    \includegraphics[width=0.32\textwidth]{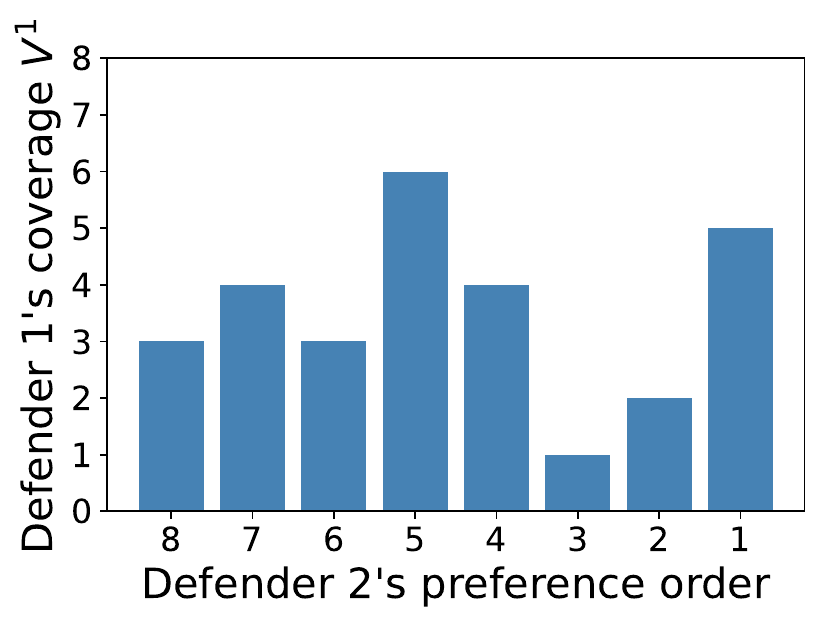}
    \includegraphics[width=0.32\textwidth]{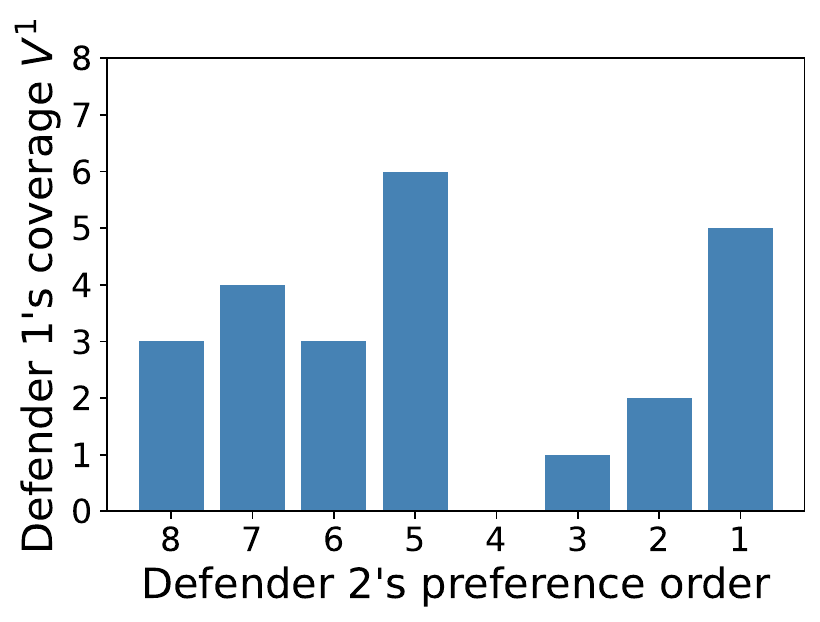}
    \includegraphics[width=0.32\textwidth]{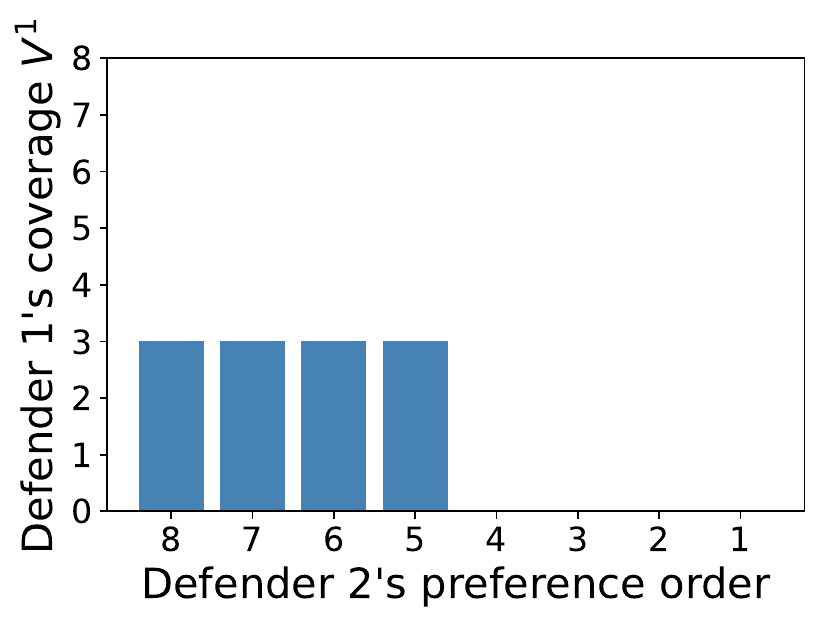}
    \caption{Left to right: reductions to obtain $t$-standard coverage for defender $1$}
    \label{fig:NSEspFigLemma}
\end{figure}
Now, for some NSE $(\mathbf{v}, t)$ consider defender $1$'s coverage profile $v^1$ based on defender $2$'s preference ordering. For illustration, assume WLOG that the targets are ordered in \textit{decreasing} order of defender $2$'s preference of attacked target. An example is shown in Figure~\ref{fig:NSEspFigLemma}(a), where $t=4$ and $T=8$, and the height of each bar corresponds to $v^1(j)$ for each target $j$. Lemma~\ref{lem:step1} simply says that the coverage profile in Figure~\ref{fig:NSEspFigLemma}(b) would also constitute an NSE. 

Suppose defender $1$'s goal is to discourage defender $2$ from deviating. Since target $4$ is attacked, defender $2$ benefits from deviating if and only if it can induce some target in $\mathcal{T}^{\succ_2}_4$ to be the new attack target. 
It does so by increasing coverage over targets in $\mathcal{T}^{\preceq_2}_4$ (potentially reducing $v^2$ elsewhere).
Conversely, defender $1$ prevents this deviation by ensuring that its minimum coverage $v^1$ over $\mathcal{T}^{\succ_2}_4$ is as high as possible.
Since $v^1$ shown in in Figure~\ref{fig:NSEspFigLemma} is part of an NSE, defender $2$ does not have a coverage $\widehat{v}^2$ such that $\widehat{v}^2(j) > 3$ for all $\mathcal{T}^{\preceq_2}_4$. If such a $\widehat{v}^2$ existed, defender $2$ could simply deviate to that (placing no coverage on targets in $\mathcal{T}^{\succ_2}_4$), This induces either target $6 \succ_2 4$ to be attacked. Hence, for defender $1$ to discourage deviations from defender $2$, it should reduce $v^1$ in all targets belonging to $\mathcal{T}^{\succ_2}_4$ to be 3, and similarly for all $\mathcal{T}^{\preceq_2}_4$, as shown in Figure~\ref{fig:NSEspFigLemma}(c). This new coverage remains a NSE.
Lemma~\ref{lem:step2} formalizes the above argument.

\begin{lemma} 
    Suppose $(v^1, v^2, t)$ is an NSE with $v^1(t) = v^2(t) = 0$. 
    Then there exists another coverage $\widetilde{v}^1$  where (i) $\widetilde{v}^1(t) = 0$, (ii) $\widetilde{v}^1(j) = 0$ for all $j \prec_2 t$, and (iii) $\widetilde{v}^1(j) = \min_{k \succ_2 t }v^1(k)$ for all $j \succ_2 t$, 
    such that $(\widetilde{v}^1, v^2, t)$ is an NSE.
    \label{lem:step2}
\end{lemma}

Lemma~\ref{lem:step2} allows us to transform any NSE $(v^1, v^2, t)$ into a new NSE by adjusting $v^1$ appropriately. Clearly, a similar process can be done for defender $2$ and $v^2$. Applying Lemma~\ref{lem:step2} for both player yields coverages with simpler structures.


\begin{definition} [$t$-standard Coverage]
    For fixed $t \in [T]$, $v^1 \in V^1$ is a $t$-standard coverage for defender $1$ if
    (i) there exists $h^1 \ge 0$ and $v^1(j) = h^1$ for all $j \succ_2 t $, and (ii) $v^1(j) = 0$ for all $j \preceq_2 t$. The same holds for defender $2$.
    \label{def:standard coverage}
\end{definition}

Reducing a coverage in an NSE $(v^1, v^2, t)$ into one containing $t$-standard coverage is done by first applying Lemma~\ref{lem:step1} to each player, followed by Lemma~\ref{lem:step2} to each player. 
Figure~\ref{fig:NSEspFigLemma} illustrates how $v^1$ evolves according to this reduction as per the prior discussions.
    Let $\mathcal{H}$ be the (possibly empty) set of NSE $(v^1, v^2, t)$ such that both $v^1$ and $v^2$ are $t$-standard coverage. 
    For a fixed $t\in [T]$, we define $\mathcal{H}_t$ to be all NSE $(v^1, v^2, t)$ in $\mathcal{H}$ where $t$ is attacked. By definition, we have (i) $\mathcal{H}_t$ and $\mathcal{H}_{t'}$ are disjoint when $t \neq t'$, and (ii) $\mathcal{H}=\bigcup_{t \in [T]} \mathcal{H}_t$.
    The existence of NSE is equivalent to saying that $\mathcal{H}$ is non-empty.

Our algorithm for computing an NSE under Assumption~\ref{ass:SSAS} is straightforward: iterate over all targets $t \in [T]$ and search for some element in $\mathcal{H}_t$. 
Finding some element (if it exists) of $\mathcal{H}_t$ for any $t$ is done in polynomial time by solving $4$ linear programs with size polynomial in $S^i$, $n$ and $T$, as shown in Algorithm~\ref{alg:basic}. 
\textsc{MaximinCov}$(\mathcal{T}, V^i)$ is an oracle finding the coverage $v^i \in V^i$ which maximizes the minimum coverage over a given set of targets $\mathcal{T}$. \textsc{MaximinCov} returns $+ \infty$ when $\mathcal{T} = \emptyset$.
When $V^i$ is defined by a finite set of schedules $\mathcal{S}^i$ and SSAS,  \textsc{MaximinCov} can be implemented via linear programming (Algorithm~\ref{alg:maxmin-cov-simple}). 

\begin{figure}[t]
\centering
\begin{minipage}[c]{0.45 \textwidth}
\centering
\begin{algorithm}[H]
	\caption{NSE for 2 defenders}
	\label{alg:basic}
	\begin{algorithmic}
		\STATE \textbf{Input: }$n$, $T$, $V^1, V^2$
		\STATE \textbf{Output: } an NSE $(\mathbf{v}, t)$ 
		\FOR{$t = 1$ \textbf{to} $T$}
		    \STATE $h^1$ $\leftarrow$ \textsc{MaximinCov}($\mathcal{T}^{\succ_2}_t$, $V^1$)
		    \STATE $g^1$ $\leftarrow$ \textsc{MaximinCov}($\mathcal{T}^{\preceq_2}_t$, $V^2$)
		    \STATE $h^2$ $\leftarrow$ \textsc{MaximinCov}($\mathcal{T}^{\succ_1}_t$, $V^2$)
		    \STATE $g^2$ $\leftarrow$ \textsc{MaximinCov}($\mathcal{T}^{\preceq_1}_t$, $V^1$)
                \IF{$h^1 \geq g^1$ and $h^2 \geq g^2$}
                    \FOR{$j = 1$ \textbf{to} $T$}
                        \STATE $v^1(j) \leftarrow h^1$ if $j \in \mathcal{T}^{\succ_2}_t$ else $0$
                        \STATE $v^2(j) \leftarrow h^2$ if $j \in \mathcal{T}^{\succ_1}_t$ else $0$
                    \ENDFOR
                    \RETURN $(v^1, v^2, t)$
                \ENDIF
		\ENDFOR
	\end{algorithmic}
\end{algorithm}
\end{minipage}
\begin{minipage}[c]{0.54 \textwidth}

\begin{algorithm}[H]
    \caption{\textsc{MaximinCov}}
    \label{alg:maxmin-cov-simple}
    \begin{align*}
    \textbf{Input: }& \text{schedules } \mathcal{S}^i, \mathcal{T}\subseteq [T] \\
    \textbf{Output: }& \text{$i$'s Maximin coverage over } \mathcal{T} \\
       \textbf{Solve: } &\max \quad h \\
        \text{s.t. } &v^i(t) \leq \sum_{s_z^i \in \mathcal{S}^i} x^i(z) \cdot s_{z}^i (t) \quad \forall t \in [T] 
     \\
     &v^i(j) \geq  h \qquad \forall j \in \mathcal{T} \\
        &\sum_{z \in [S^i]} x^i(z) = 1, \quad x^i \in \mathbb{R}^{S^i}_+
    \end{align*}
    \end{algorithm}
\end{minipage}
\caption{Left: Solving NSE in $2$ defender games. Right: A linear program finding the Maximin coverage for defender $i$ when $V^i$ is given by schedules $\mathcal{S}^i$ and SSAS.}
\label{fig:algo-basic}
\end{figure}
\begin{theorem}
    Under the SSAS assumption (Assumption~\ref{ass:SSAS}), $\mathcal{H} \neq \emptyset$, i.e., an NSE always exists. Consequently, Algorithm~\ref{alg:basic} is guaranteed to return an NSE. 
    \label{thm:existence-2p}
\end{theorem}

\subsection{Efficiency of NSE}
    Blindly applying Algorithm~\ref{alg:basic} can lead to a pathology where the attacked target is undesirable for both players, i.e., a NSE $(\mathbf{v}, t)$ can have $v^1(\widetilde{t}) \geq 0$ and $v^2(\widetilde{t}) \geq 0$ when $\widetilde{t} \succ_1 t$ \textit{and} $\widetilde{t} \succ_2 t$.
    For example, suppose $1 \succ_1 2 \succ_1 3$ and $1 \succ_2 2 \succ_2 3$ where
    defenders possess identity schedules, i.e., 
    $\mathcal{S}^i = \{ (1, 0, 0), (0, 1, 0), (0, 0, 1) \}$.
    Setting $\widetilde{v}^1=\widetilde{v}^2=(1,0,0)$,
    we find that $(\mathbf{\widetilde{v}}, 2)$ is an NSE as it is AIC and neither defender can give target $1$ a coverage strictly greater than $1$ by deviating.
    However, both defenders prefer target $1$ to be attacked. Indeed, a trivial NSE is $(\mathbf{v}, 1)$ where $v^1=v^2=(0, 0, 0)$, i.e., the profile $(\mathbf{v}, 1)$ ``Pareto dominates'' $(\mathbf{\widetilde{v}}, 2)$. 

\begin{definition}[Inefficiency] An NSE $(v^1, v^2, t) \in \mathcal{H}_t$ is inefficient (resp. efficient) if and only if there exists  $j$ where $j \succ_1 t$ and $j \succ_2 t$. 
Targets constituting inefficient (resp. efficient) NSE are called inefficient targets.
\label{def:ineff_NSE}
\end{definition}



Targets where $\mathcal{H}_t=\emptyset$ are neither efficient nor inefficient. 
There may exist multiple efficient NSE. However, we show that efficient NSE exists under SSAS.

\begin{lemma}
    Let $j, t \in [T]$ such that $\mathcal{H}_t\neq \emptyset$, $j \succ_1 t$ and $j \succ_2 t$. Then, $\mathcal{H}_j\neq \emptyset$.
    \label{lem: no same cover}
\end{lemma}

\subsection{Exploiting Additional Structure in $V^i$}
We now move towards characterizing and solving for equilibrium in two classes of games which contain even more structure in $V^i$. 
\subsubsection{Finding NSE under Monotone Schedules}
The first is the case of \textit{Monotone Schedules}, where each defender's schedule places less coverage on targets preferred to be attacked.
This allows us to efficiently solve for NSE when $n\geq 2$.
\begin{assumption}[Monotone schedules (MS)]
    A schedule $s^i_z \in \mathcal{S}^i$ 
    is Monotone if $j \succ_i t \implies s^i_z(j) \leq s^i_z(t)$. The game possesses Monotone Schedules (MS) if all defender schedules  are monotone.
\end{assumption}
\begin{theorem}
    Under both MS and SSAS, a NSE exists for $n \geq 2$. It can be computed in polynomial time in the number of schedules, $n$, and $T$. 
    \label{thm:multiNSEexist}
\end{theorem}

\begin{proof}

Define $\pi_i(t) \in [T]$ such that target $t$ is the $\pi_i(t)$-th preferred target of defender $i$. 
Define a matrix $\mathcal{M}\in \mathbb{R}^{n\times T}$, where $\mathcal{M}_{i, \pi_i(j)} = \textsc{MaximinCov}(\mathcal{T}_{j}^{\preceq_i}, V^i)$ for each defender $i$ and target $j$.
\footnote{i.e., for each defender, reorder targets from most to least preferred (Figure~\ref{fig:NSEspFigLemma}) and compute the maximin coverage for targets comprising $t$ and everything less preferred.}
$\mathcal{M}$ is non-decreasing from left to right. 
Let $F(t) = \max_{i} \mathcal{M}_{i, \pi_i(t)}$. 
For convenience, we assume $F(t)\neq F(t')$ for $t\neq t'$ such that $\Argmin_{t} F(t)=\{ k^* \}$. 
We construct a NSE $(\mathbf{v}, k^*)$ where: 
\begin{enumerate}
\item for all defenders $i$, $v^i(k^*) = 0$ such that $v^{total}(k^*)=0$, and
\item for every $t \neq k^*$, we find a defender $i$ where $\mathcal{M}_{i, \pi_i(t)} = F(t)$. We set coverage $v^i(t) = F(k^*)$ and $v^{i'}(t) = 0$ for all $i' \in [n] \backslash \{i\}$.
\end{enumerate}
Clearly, the above algorithm runs polynomial time. 
In Step 2, we have by definition at least one defender $i$ satisfying $\mathcal{M}_{i, \pi_i(t)} = F(t)$. We first show that $\mathbf{v}$ is achievable, i.e., $v^i \in V^i$ for all defenders $i$. 
Fix $i \in [n]$. Let $\mathcal{T}$ be the set of targets covered by it, each with coverage $F(k^*)$. By Step 1, $k^*\notin \mathcal{T}$. Furthermore, $\mathcal{T} \subseteq \mathcal{T}_{k^*}^{\prec_i}$. If target $j\in \mathcal{T}$ and $j \succ_i k^*$, then $\mathcal{M}_{i, \pi_i(j)}\le \mathcal{M}_{i, \pi_i(k^*)}$ since $\mathcal{M}$ is non-decreasing. This contradicts $\mathcal{M}_{i, \pi_i(j)} = F(j) > F(k^*) \ge \mathcal{M}_{i, \pi_i(k^*)}$ from the definition of $F$ and $k^*$.
Now, let $j^*$ be defender $i$'s most preferred target in $\mathcal{T}$. 
Consider a coverage $\widetilde{v}^i$ with $\widetilde{v}^i(j) = F(j^*)$ for $j \preceq_i j^*$ and $\widetilde{v}^i(j) = 0$ otherwise. 
In Step 2, $v^i(j^*) > 0$ implies $F(j^*) = \mathcal{M}_{i, \pi_i(j^*)}$, because only targets $j$ such that $F(j) = \mathcal{M}_{i, \pi_i(j)}$ are covered by $i$. Hence, $F(j^*) = \mathcal{M}_{i, \pi_i(j^*)} = \textsc{MaximinCov}(\mathcal{T}_{j^*}^{\preceq_i}, V^i)$ and 
$\widetilde{v}^i \in V^i$. 
Since $F(j^*)\ge F(k^*)$, $\widetilde{v}^i(j) \geq v^i(j)$ for $j \in [T]$, i.e., coverage of $\widetilde{v}^i$ is no less than $v^i$. Because $\widetilde{v}^i \in V^i$, $v^i\in V^i$ by SSAS.


Lastly, we show that $(\mathbf{v}, k^*)$ is indeed an NSE.
It is AIC since $v^{total}(k^*) = 0$. 
Next, we prove that for any defender $i$, any $\widehat{v}^i\in V^i$ and any target $t \succ_i k^*$, $(v^1, \cdots, \widehat{v}^i, \cdots, v^n, t)$ is not $i$-WAIC. 
First, we show that $\widehat{v}^{total}(t) \ge F(k^*)$. We have
$\widehat{v}^{total}(t) = \widehat{v}^i(t) + \sum_{i'\neq i}v^{i'}(t) \ge \sum_{i'\neq i}v^{i'}(t)$. 
Since $v^i$ has no coverage on $\mathcal{T}_{k^*}^{\succ_i}$, $v^i(t) = 0$. Therefore, $\sum_{i'\neq i}v^{i'}(t) = v^{total}(t) = F(k^*)$. We have $\widehat{v}^{total}(t) \ge \sum_{i'\neq i}v^{i'}(t) = F(k^*)$. Second, $\widehat{v}^{total}(k^*)\le F(k^*)$. $\widehat{v}^{total}(k^*) = \widehat{v}^{i}(k^*)$ since $v^{i'}(k^*) = 0$ for any $i'\neq i$. We claim that $\widehat{v}^{i}(k^*) \le F(k^*)$. If $\widehat{v}^{i}(k^*) > F(k^*)$, defender $i$ can cover all targets in $\mathcal{T}_{k^*}^{\preceq_i}$ with $\widehat{v}^{i}(k^*) > F(k^*)$ by MS. Therefore, $\mathcal{M}_{i, \pi_i(k^*)} > F(k^*)$, which conflicts with $F(k^*) = \max_{i'} \mathcal{M}_{i', \pi_{i'}(k^*)}$. We have $\widehat{v}^{total}(k^*) = \widehat{v}^{i}(k^*) \le F(k^*)$.
In conclusion, $\widehat{v}^{total}(t) \ge \widehat{v}^{total}(k^*)$. If $t\in B(\widehat{v}^{total})$, $k^*\in B(\widehat{v}^{total})$ holds, so $(v^1, \cdots, \widehat{v}^i, \cdots, v^n, t)$ is not $i$-WAIC. $(\mathbf{v}, k^*)$ is $i$-IC. \hfill$\blacksquare$



\end{proof}

\subsubsection{Efficient Solvers when $V^i$ is Compactly Represented }
\label{sec:efficient-compact}

In Algorithm~\ref{alg:basic}, we were required to optimize over $t$-standard coverage for each defender. This involves solving linear programs.
Unfortunately, the number of schedules can be prohibitively large. For example, in patrolling on layered networks, the number of schedules is exponential in its depth and is computationally infeasible for large games.
Fortunately, both our proof of existence and algorithm operate in the space of attainable coverage $V^i$ and not directly on $x$ and $S^i$. 
In our example, $V^i$ can be expressed as flows in the network (and more generally, any directed acyclic graph with a source and a sink), which in turn is a polyhedron with a polynomial number of constraints (in terms of the number of edges and vertices). 

\section{Experiments}
Our experiments are conducted on an Intel(R) Core(TM) i7-7700K CPU @ 4.20GHz. We use Gurobi \cite{gurobi} to solve linear programs. We seek to answer the following.
    (i) Can NSE be practically computed for reasonable environments? How does computational time scale with parameters such as $T$, $S^i$, $V^i$, and $n$?
    (ii) How does an NSE look like qualitatively? When $n=2$, how many NSE are efficient? What proportion of targets are included in \textit{some} NSE? 
    (iii) What is the quality (in terms of the attacked target) in the multiple-defender setting as compared to single defender settings?
We explore 3 synthetically generated games, where defender preferences $\succ_i$ are generated uniformly at random.
\begin{itemize}
\item \textit{Randomly Generated Schedules (RGS)}.
We generate games with random schedules where each $s^i_j$ is a random integer from $[0, 10]$, and the number of schedules $S^i$, $T$ are specified. In some cases, we limit each schedule's support size to be smaller than $T$ (with the support again randomly selected). 
\item \textit{Public Security on Grids (PSG)}.
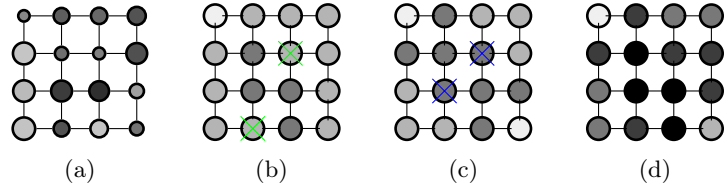
\begin{figure}[t]
\centering
\begin{subfigure}[t]{.20 \textwidth}
\centering
\begin{tikzpicture}[every node/.style={circle, inner sep=0pt}]
  \pgfmathsetseed{142857}%
  \foreach \x in {1,...,4} {
    \foreach \y in {1,...,4} {
      \pgfmathsetmacro{\nodesize}{1.5 + 1.5*rnd} 
      \pgfmathsetmacro{\filldarkness}{20 + 200*rnd} 
      \definecolor{nodecolor}{RGB}{\filldarkness, \filldarkness, \filldarkness}
      \node[draw, line width=1 pt, fill=nodecolor, minimum size=\nodesize mm] (\x-\y) at (0.5*\x, 0.5*\y) {};
    }
  }

  \foreach \x in {1,...,4} {
    \foreach \y in {1,...,4} {
      \ifnum\x<4
        \draw (\x-\y) -- (\the\numexpr\x+1\relax-\y);
      \fi
      
      \ifnum\y>1
        \draw (\x-\y) -- (\x-\the\numexpr\y-1\relax);
      \fi
    }
  }
\end{tikzpicture}
\caption{}
\label{fig:public-sec-population}
\end{subfigure}
\begin{subfigure}[t]{.20 \textwidth}
\centering
\begin{tikzpicture}[every node/.style={circle, inner sep=0pt}]
  \tikzset{cross/.style={cross out, draw, 
         minimum size=2*(#1-\pgflinewidth), 
         inner sep=0pt, outer sep=0pt}}
\definecolor{C4}{RGB}{0, 0, 0}
\definecolor{C3}{RGB}{60, 60, 60}
\definecolor{C2}{RGB}{120, 120, 120}
\definecolor{C1}{RGB}{180, 180, 180}
\definecolor{C0}{RGB}{240, 240, 240}

\node[draw, fill=C1, line width=1 pt, minimum size=3 mm] (0.5-0.5) at (0.5, 0.5) {};
\node[draw, fill=C1, line width=1 pt, minimum size=3 mm] (1.0-0.5) at (1.0, 0.5) {};
\node[draw, fill=C2, line width=1 pt, minimum size=3 mm] (1.5-0.5) at (1.5, 0.5) {};
\node[draw, fill=C1, line width=1 pt, minimum size=3 mm] (2.0-0.5) at (2.0, 0.5) {};

\node[draw, fill=C1, line width=1 pt, minimum size=3 mm] (0.5-1.0) at (0.5, 1.0) {};
\node[draw, fill=C2, line width=1 pt, minimum size=3 mm] (1.0-1.0) at (1.0, 1.0) {};
\node[draw, fill=C2, line width=1 pt, minimum size=3 mm] (1.5-1.0) at (1.5, 1.0) {};
\node[draw, fill=C1, line width=1 pt, minimum size=3 mm] (2.0-1.0) at (2.0, 1.0) {};

\node[draw, fill=C1, line width=1 pt, minimum size=3 mm] (0.5-1.5) at (0.5, 1.5) {};
\node[draw, fill=C2, line width=1 pt, minimum size=3 mm] (1.0-1.5) at (1.0, 1.5) {};
\node[draw, fill=C1, line width=1 pt, minimum size=3 mm] (1.5-1.5) at (1.5, 1.5) {};
\node[draw, fill=C1, line width=1 pt, minimum size=3 mm] (2.0-1.5) at (2.0, 1.5) {};

\node[draw, fill=C0, line width=1 pt, minimum size=3 mm] (0.5-2.0) at (0.5, 2.0) {};
\node[draw, fill=C1, line width=1 pt, minimum size=3 mm] (1.0-2.0) at (1.0, 2.0) {};
\node[draw, fill=C1, line width=1 pt, minimum size=3 mm] (1.5-2.0) at (1.5, 2.0) {};
\node[draw, fill=C1, line width=1 pt, minimum size=3 mm] (2.0-2.0) at (2.0, 2.0) {};

  \draw (1.5,1.5) node[cross=5pt,rotate=0,green]{};
  \draw (1.0,0.5) node[cross=5pt,rotate=0,green]{};
  \foreach \x in {1,...,4} {
    \foreach \y in {1,...,4} {
      \ifnum\x<4
        \draw (\x-\y) -- (\the\numexpr\x+1\relax-\y);
      \fi
      
      \ifnum\y>1
        \draw (\x-\y) -- (\x-\the\numexpr\y-1\relax);
      \fi
    }
  }
\end{tikzpicture}
\caption{}
\label{fig:public-sec-police-coverage}
\end{subfigure}
\begin{subfigure}[t]{.20 \textwidth}
\centering
\begin{tikzpicture}[every node/.style={circle, inner sep=0pt}]
  \tikzset{cross/.style={cross out, draw, 
         minimum size=2*(#1-\pgflinewidth), 
         inner sep=0pt, outer sep=0pt}}

\definecolor{C4}{RGB}{0, 0, 0}
\definecolor{C3}{RGB}{60, 60, 60}
\definecolor{C2}{RGB}{120, 120, 120}
\definecolor{C1}{RGB}{180, 180, 180}
\definecolor{C0}{RGB}{240, 240, 240}

\node[draw, fill=C1, line width=1 pt, minimum size=3 mm] (0.5-0.5) at (0.5, 0.5) {};
\node[draw, fill=C1, line width=1 pt, minimum size=3 mm] (1.0-0.5) at (1.0, 0.5) {};
\node[draw, fill=C2, line width=1 pt, minimum size=3 mm] (1.5-0.5) at (1.5, 0.5) {};
\node[draw, fill=C0, line width=1 pt, minimum size=3 mm] (2.0-0.5) at (2.0, 0.5) {};

\node[draw, fill=C1, line width=1 pt, minimum size=3 mm] (0.5-1.0) at (0.5, 1.0) {};
\node[draw, fill=C2, line width=1 pt, minimum size=3 mm] (1.0-1.0) at (1.0, 1.0) {};
\node[draw, fill=C2, line width=1 pt, minimum size=3 mm] (1.5-1.0) at (1.5, 1.0) {};
\node[draw, fill=C2, line width=1 pt, minimum size=3 mm] (2.0-1.0) at (2.0, 1.0) {};

\node[draw, fill=C2, line width=1 pt, minimum size=3 mm] (0.5-1.5) at (0.5, 1.5) {};
\node[draw, fill=C2, line width=1 pt, minimum size=3 mm] (1.0-1.5) at (1.0, 1.5) {};
\node[draw, fill=C2, line width=1 pt, minimum size=3 mm] (1.5-1.5) at (1.5, 1.5) {};
\node[draw, fill=C1, line width=1 pt, minimum size=3 mm] (2.0-1.5) at (2.0, 1.5) {};

\node[draw, fill=C0, line width=1 pt, minimum size=3 mm] (0.5-2.0) at (0.5, 2.0) {};
\node[draw, fill=C2, line width=1 pt, minimum size=3 mm] (1.0-2.0) at (1.0, 2.0) {};
\node[draw, fill=C1, line width=1 pt, minimum size=3 mm] (1.5-2.0) at (1.5, 2.0) {};
\node[draw, fill=C1, line width=1 pt, minimum size=3 mm] (2.0-2.0) at (2.0, 2.0) {};

  \draw (1.5,1.5) node[cross=5pt,rotate=0,blue]{};
  \draw (1.0,1.0) node[cross=5pt,rotate=0,blue]{};
  \foreach \x in {1,...,4} {
    \foreach \y in {1,...,4} {
      \ifnum\x<4
        \draw (\x-\y) -- (\the\numexpr\x+1\relax-\y);
      \fi
      
      \ifnum\y>1
        \draw (\x-\y) -- (\x-\the\numexpr\y-1\relax);
      \fi
    }
  }
\end{tikzpicture}
\caption{}
\label{fig:public-sec-VIP-coverage}
\end{subfigure}
\begin{subfigure}[t]{.20 \textwidth}
\centering
\begin{tikzpicture}[every node/.style={circle, inner sep=0pt}]
  \tikzset{cross/.style={cross out, draw, 
         minimum size=2*(#1-\pgflinewidth), 
         inner sep=0pt, outer sep=0pt}}

\definecolor{C4}{RGB}{0, 0, 0}
\definecolor{C3}{RGB}{60, 60, 60}
\definecolor{C2}{RGB}{120, 120, 120}
\definecolor{C1}{RGB}{180, 180, 180}
\definecolor{C0}{RGB}{240, 240, 240}

\node[draw, fill=C2, line width=1 pt, minimum size=3 mm] (0.5-0.5) at (0.5, 0.5) {};
\node[draw, fill=C3, line width=1 pt, minimum size=3 mm] (1.0-0.5) at (1.0, 0.5) {};
\node[draw, fill=C4, line width=1 pt, minimum size=3 mm] (1.5-0.5) at (1.5, 0.5) {};
\node[draw, fill=C1, line width=1 pt, minimum size=3 mm] (2.0-0.5) at (2.0, 0.5) {};

\node[draw, fill=C2, line width=1 pt, minimum size=3 mm] (0.5-1.0) at (0.5, 1.0) {};
\node[draw, fill=C4, line width=1 pt, minimum size=3 mm] (1.0-1.0) at (1.0, 1.0) {};
\node[draw, fill=C4, line width=1 pt, minimum size=3 mm] (1.5-1.0) at (1.5, 1.0) {};
\node[draw, fill=C3, line width=1 pt, minimum size=3 mm] (2.0-1.0) at (2.0, 1.0) {};

\node[draw, fill=C3, line width=1 pt, minimum size=3 mm] (0.5-1.5) at (0.5, 1.5) {};
\node[draw, fill=C4, line width=1 pt, minimum size=3 mm] (1.0-1.5) at (1.0, 1.5) {};
\node[draw, fill=C3, line width=1 pt, minimum size=3 mm] (1.5-1.5) at (1.5, 1.5) {};
\node[draw, fill=C3, line width=1 pt, minimum size=3 mm] (2.0-1.5) at (2.0, 1.5) {};

\node[draw, fill=C0, line width=1 pt, minimum size=3 mm] (0.5-2.0) at (0.5, 2.0) {};
\node[draw, fill=C3, line width=1 pt, minimum size=3 mm] (1.0-2.0) at (1.0, 2.0) {};
\node[draw, fill=C2, line width=1 pt, minimum size=3 mm] (1.5-2.0) at (1.5, 2.0) {};
\node[draw, fill=C2, line width=1 pt, minimum size=3 mm] (2.0-2.0) at (2.0, 2.0) {};

  \foreach \x in {1,...,4} {
    \foreach \y in {1,...,4} {
      \ifnum\x<4
        \draw (\x-\y) -- (\the\numexpr\x+1\relax-\y);
      \fi
      
      \ifnum\y>1
        \draw (\x-\y) -- (\x-\the\numexpr\y-1\relax);
      \fi
    }
  }
\end{tikzpicture}
\caption{}
\label{fig:public-sec-total-coverage}
\end{subfigure}
\caption{Public Security Game with $m=4$, $r=2$, using the L1-distance.
(a) Preference profiles for defenders. Size of nodes represent relative values of targets for each defender.
(b) Coverage by the defender $1$, $v^1$ by evenly distributing resource between  buildings marked by in green. Darker targets enjoy a higher coverage. Note that overlapping regions get double the coverage.
(c) Same as (b), but for defender $2$.
(d) Total coverage $v^{total}$ accrued over both entities. 
}
\label{fig:public-sec}
\end{figure}
An event is held in the streets of Manhattan, which we abstract as a $m \times m$ grid (Figure~\ref{fig:public-sec}), where each vertex represents a target (building).
The security for the event is managed by two entities: the city police and VIP security detail, each with different priorities. 
Each places checkpoints distributed across buildings in the city, which provide coverage in a radius $r$. The level of coverage is dependent on the average number of officers allocated to it. 
An example with $m=4, r=2$ is shown in Figure~\ref{fig:public-sec}. 

\item \textit{Patrolling on Layered Networks (PLN)}.
We follow the motivating example in Section~\ref{sec:background}, varying the width and number of layers. Each patrol can only change its ``level'' (position on the y-axis) by at most one step between layers. Unlike the public security game, there are now an exponential number of paths, hence computational costs become an important consideration as well.
\end{itemize}

\subsection{Computational Costs of Computing NSE}
We first restrict ourselves to $2$-defender settings under SASS. We evaluate computational efficiency of the algorithms of Figure~\ref{fig:algo-basic}, with results for RGS, PSG and PLN summarized in Figures~\ref{fig:running_time_two_defender} and \ref{fig:grid_time-and-network_time}. For PLN, we utilized the efficient method in Section~\ref{sec:efficient-compact}.
We ran 100 trials for each parameter setting and report the mean computation time and their standard errors (which were mostly negligible).

In RGS, we varied $T$, $S^i$ and the support size of each schedule. As expected, the average running time increases superlinearly with $T$ (Figure~\ref{fig:expt-runtime-rgs-T}), since the loop in Algorithm~\ref{alg:basic} is repeated more times, and calls to \textsc{MaximinCov} also incur a higher cost. However, Figure~\ref{fig:expt-runtime-rgs-S} shows that as $S^i$ increases, the required running time increases linearly. This is unexpected, since $S^i$ is involved in the \textsc{MaximinCov} subroutine, whose constraint matrix grows linearly with $S^i$. This suggests that the runtime of \textsc{MaximinCov} grows linearly with $S^i$, atypical of linear programs. This could be because (i) our problems are small by standards of modern solvers, or (ii) the solver exploits additional structure under the hood. Lastly, Figure~\ref{fig:expt-runtime-rgs-Supp} shows that adjusting support size does not impact running time. This is unsurprising since the solver is not explicitly told to exploit sparsity.

\begin{figure}[t]
	\centering

    \begin{subfigure}[t]{0.32 \textwidth}
            \centering
    	  \includegraphics[width=\linewidth]{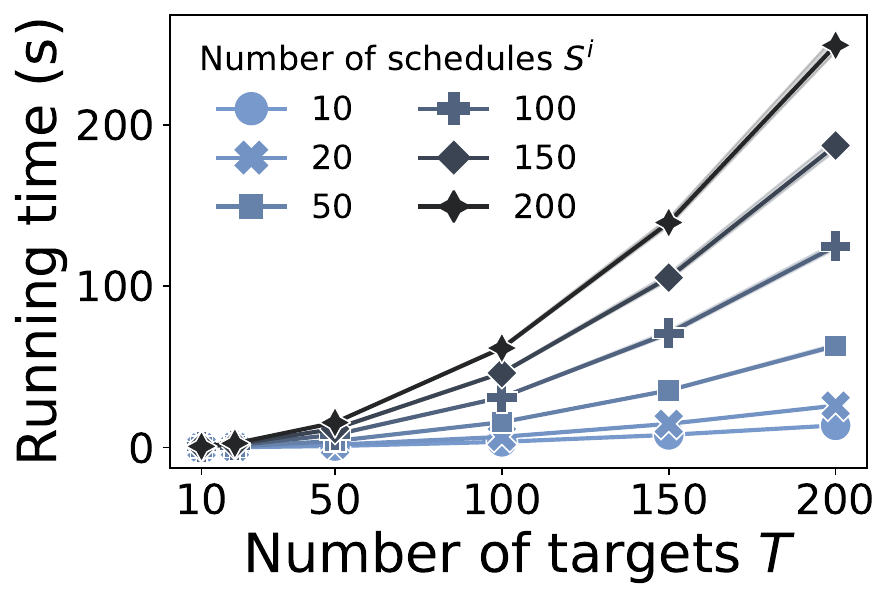}
            \caption{Runtime as $T$ varies}
            \label{fig:expt-runtime-rgs-T}
    \end{subfigure}
    \begin{subfigure}[t]{0.32 \textwidth}
    \centering
    	  \includegraphics[width=\linewidth]{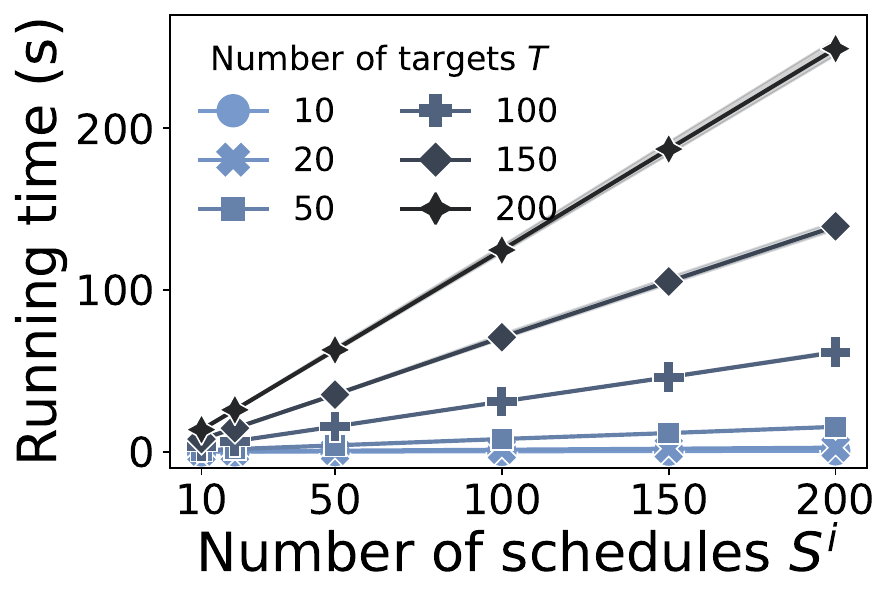}
            \caption{Runtime as $S^i$ varies}
            \label{fig:expt-runtime-rgs-S}
    \end{subfigure}
    \begin{subfigure}[t]{0.32 \textwidth}
    \centering
        \includegraphics[width=\linewidth]{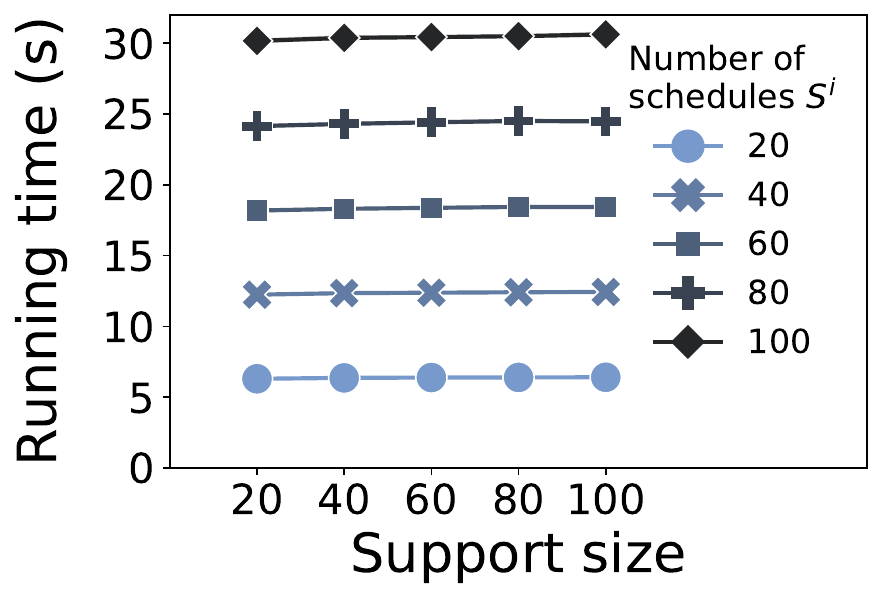}
        \caption{Runtime as support size varies, $T=100$}
            \label{fig:expt-runtime-rgs-Supp}
    \end{subfigure}
	\caption{Wallclock time to find one efficient NSE using the algorithm in Figure~\ref{fig:algo-basic}.}
	\label{fig:running_time_two_defender}
\end{figure}

In PSG, we varied $T$ by increasing the grid size $m$ from $4$ to $10$. As with RGS, Figure~\ref{fig:expt-grid-T} running time superlinearly with $T$. We also indirectly adjusted the support size by adjusting $r$, the radius of security coverage (Figure~\ref{fig:expt-grid-r}). Once again, we did not notice any appreciable difference in running times. Similarly for PLN, we note a superlinear growth in running time as the network enlarges, be it from increasing layers or width of the network (Figure~\ref{fig:expt-network-L} and \ref{fig:expt-network-w}).

\begin{figure}[t]
	\centering
 \begin{subfigure}[t]{0.32 \textwidth}
 \centering
    	  \includegraphics[width=\linewidth]{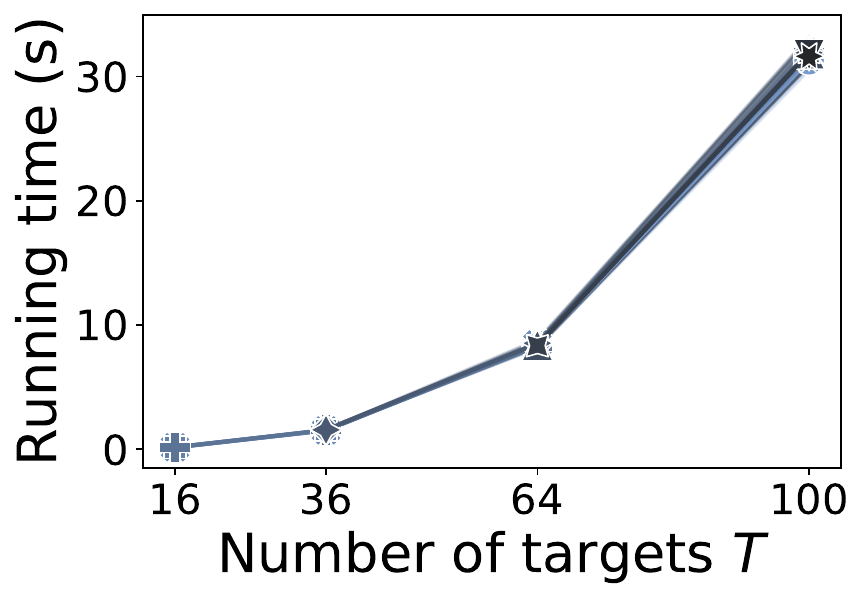}
       \caption{Runtime as $T$ varies}
       \label{fig:expt-grid-T}
       \end{subfigure}
       \begin{subfigure}[t]{0.32 \textwidth}
       \centering
    	  \includegraphics[width=\linewidth]{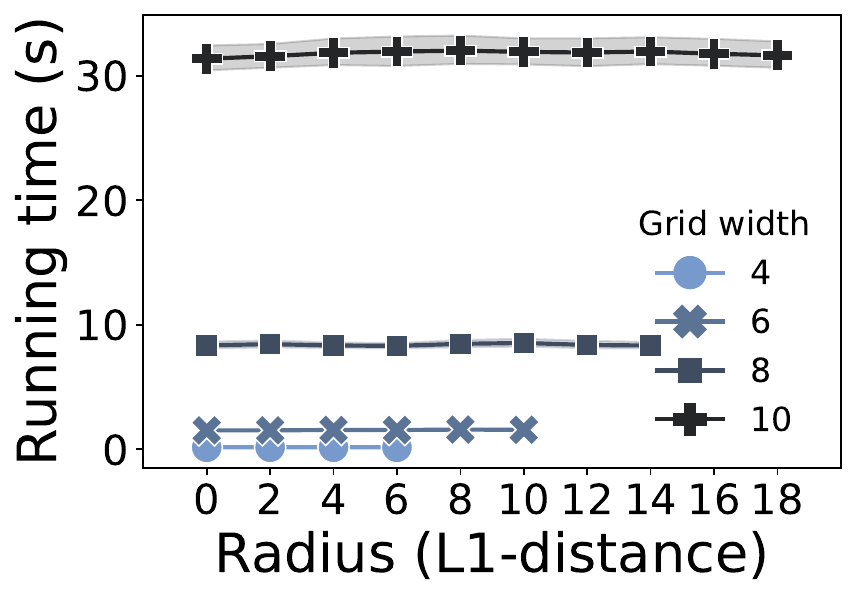}
       \caption{Runtime as $r$ varies}
       \label{fig:expt-grid-r}
       \end{subfigure}
\\
\begin{subfigure}[t]{0.32 \textwidth}
 \centering
     	  \includegraphics[width=\linewidth]{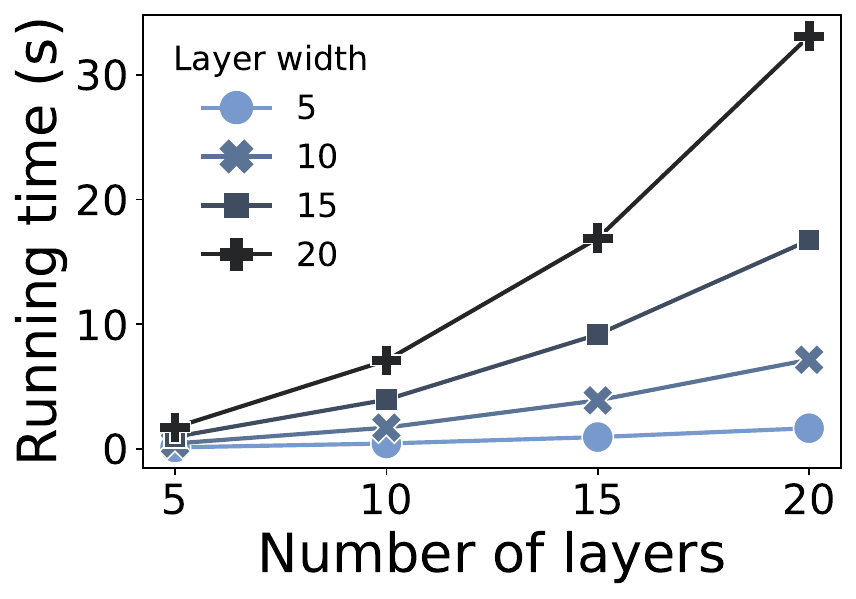}
\caption{Runtime as layers vary}
       \label{fig:expt-network-L}
\end{subfigure}
\begin{subfigure}[t]{0.32 \textwidth}
    	  \includegraphics[width=\linewidth]{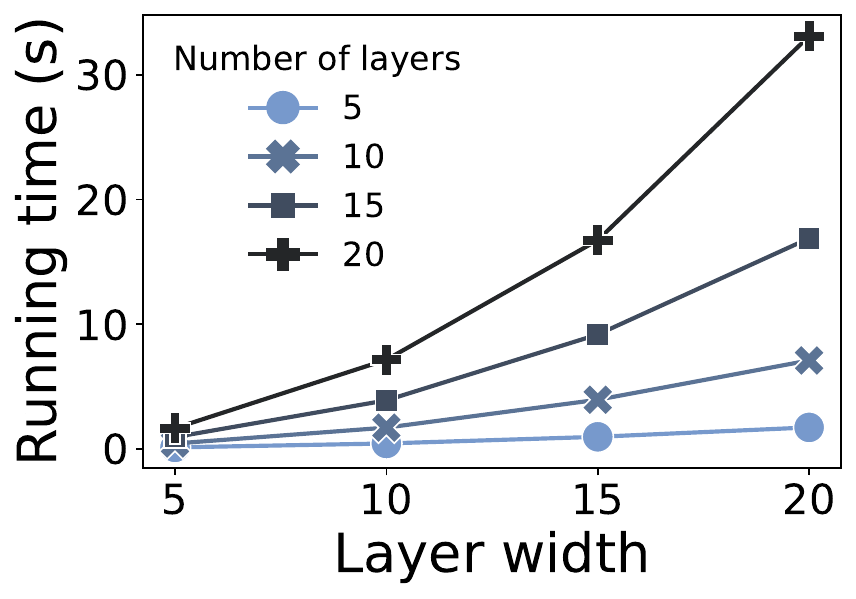}
    \caption{Runtime as width varies}
       \label{fig:expt-network-w}
\end{subfigure}
	\caption{Top: Wallclock time to compute NSE for PSG using the algorithm in Figure~\ref{fig:algo-basic}. Bottom: Computing NSE for PLN using the method in Section~\ref{sec:efficient-compact}.}
	\label{fig:grid_time-and-network_time}
\end{figure}





We now examine multiple defenders in RGS under SSAS and MSS. Again, we ran $100$ rounds for each scenario and report the means (standard errors were negligible) in Figure~\ref{fig:multi_defender}. We observe running times increasing linearly with $n$ and schedules (omitted due to space constraints), but superlinearly with $T$. 


\begin{figure}[t]
	\centering
 \begin{subfigure}[t]{0.32 \textwidth}
    	  \includegraphics[width=\linewidth]{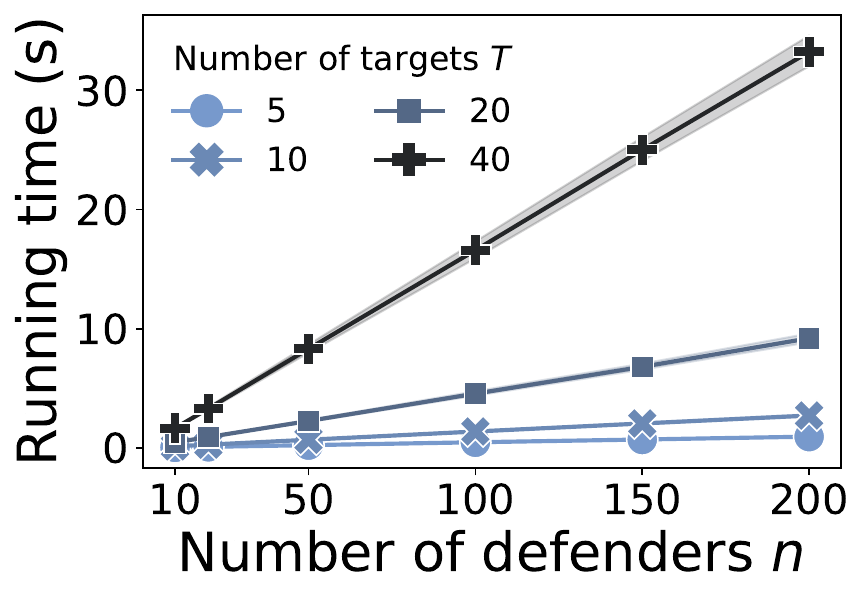}
\end{subfigure}
\begin{subfigure}[t]{0.32 \textwidth}
\includegraphics[width=\linewidth]{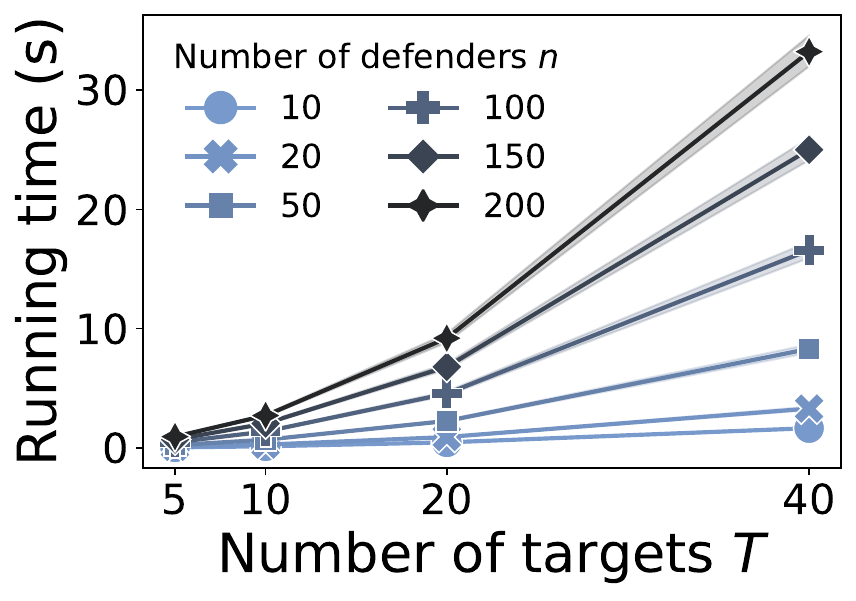}
\end{subfigure}
\begin{subfigure}[t]{0.32 \textwidth}
\includegraphics[width=\linewidth]{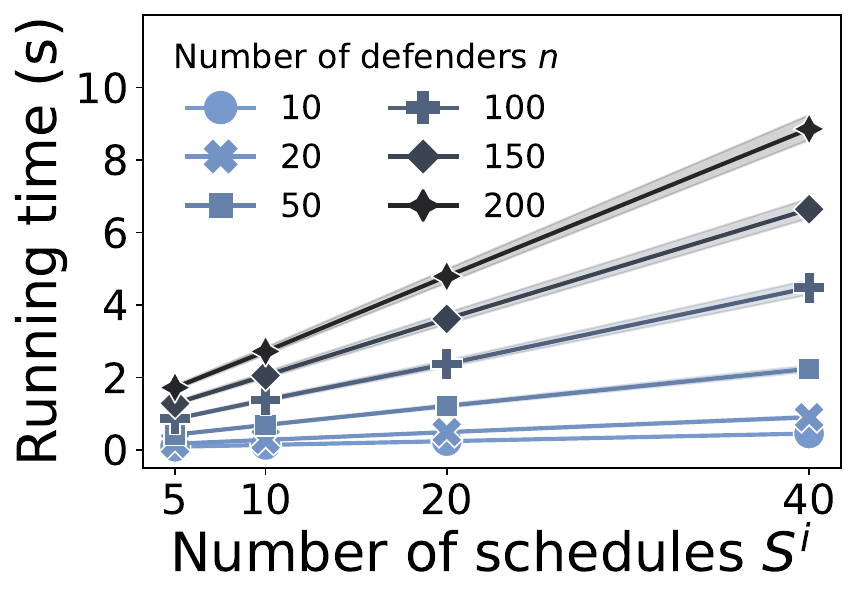}
\end{subfigure}
	\caption{Top: Running time for RGS under MSS assumptions. Top: Wallclock time as $n$ increases. Bottom: Running time as $T$ and $S^i$ vary.}
	\label{fig:multi_defender}
\end{figure}


\begin{figure}[t]
	\centering
\begin{subfigure}[t]{0.32 \textwidth}
    	  \includegraphics[width=\linewidth]{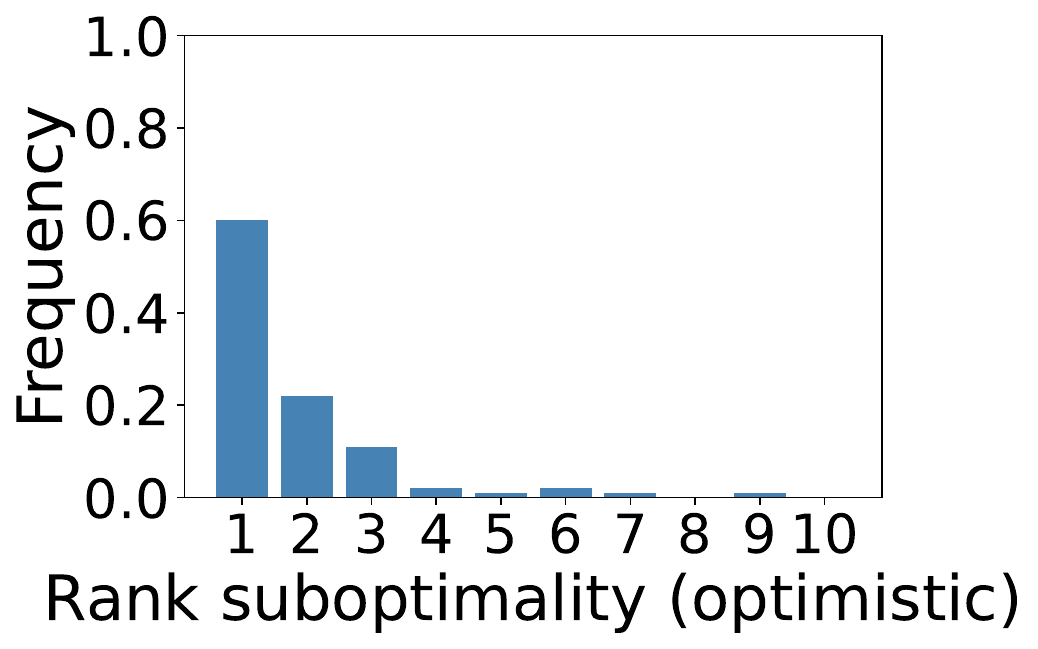}
\end{subfigure}
\begin{subfigure}[t]{0.32 \textwidth}
    	  \includegraphics[width=\linewidth]{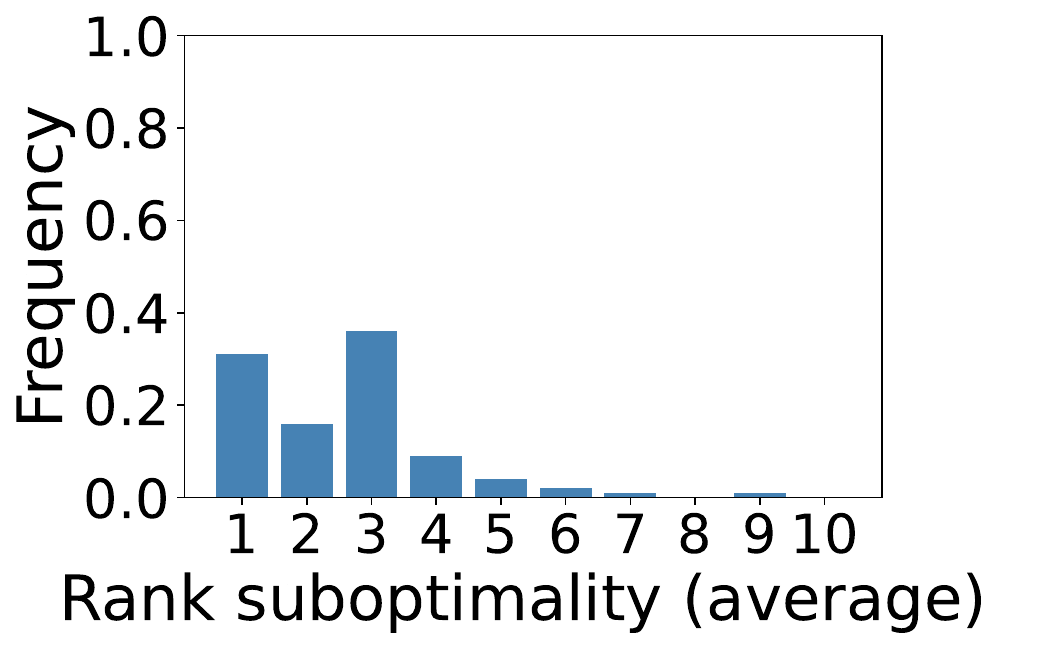}
\end{subfigure}
\begin{subfigure}[t]{0.32 \textwidth}
       \includegraphics[width=\linewidth]{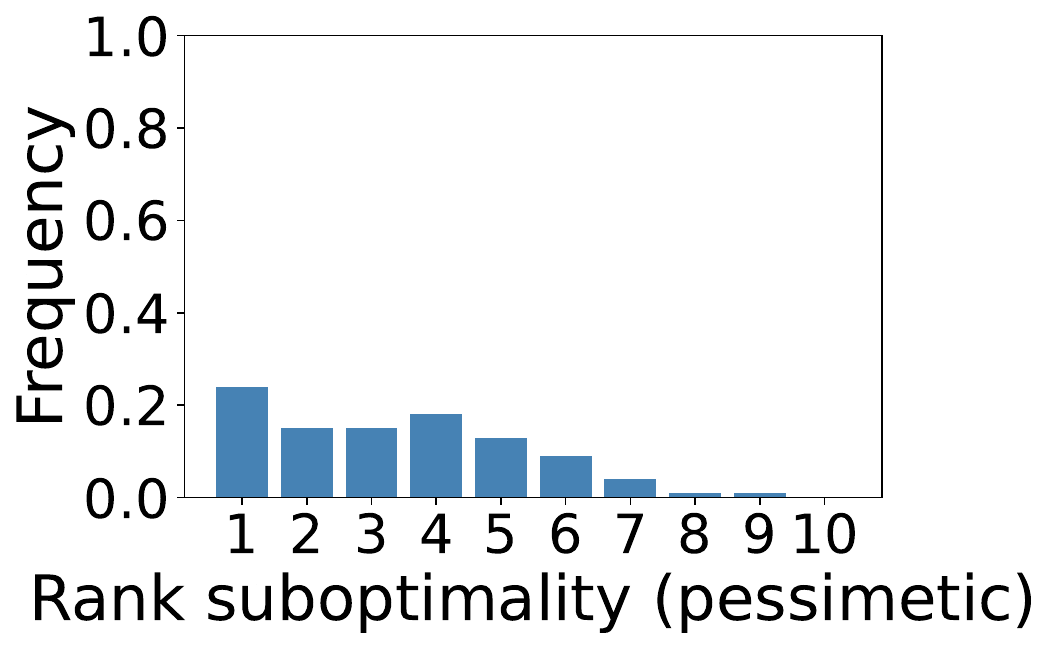}
\end{subfigure}\\
\begin{subfigure}[t]{0.32 \textwidth}
   \includegraphics[width=\linewidth]{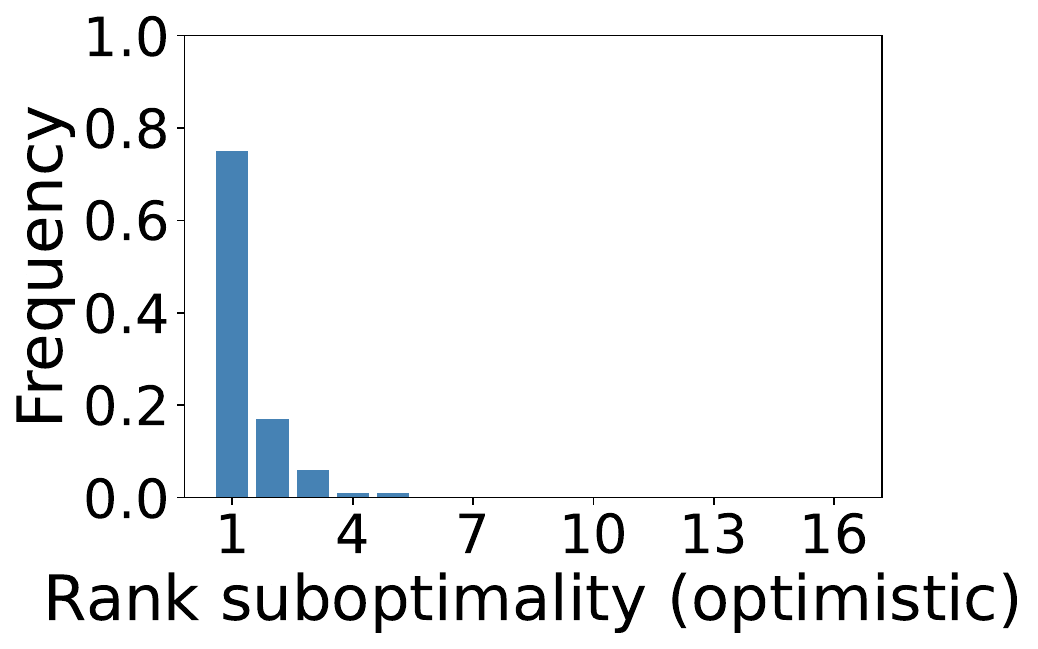}
\end{subfigure}
\begin{subfigure}[t]{0.32 \textwidth}
\includegraphics[width=\linewidth]{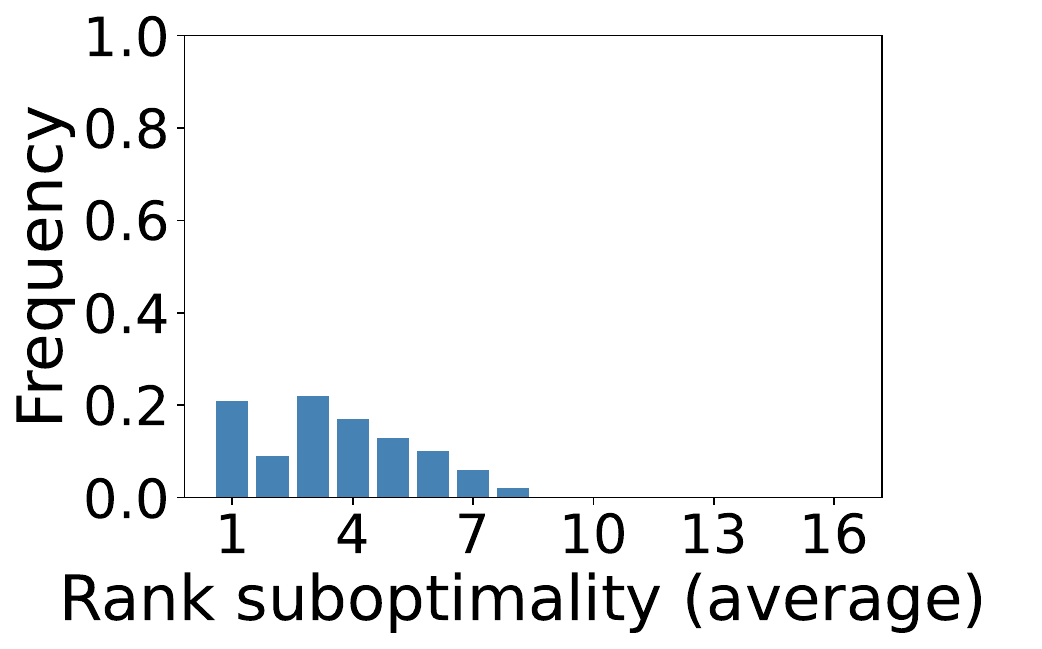}
\end{subfigure}
\begin{subfigure}[t]{0.32 \textwidth}
       \includegraphics[width=\linewidth]{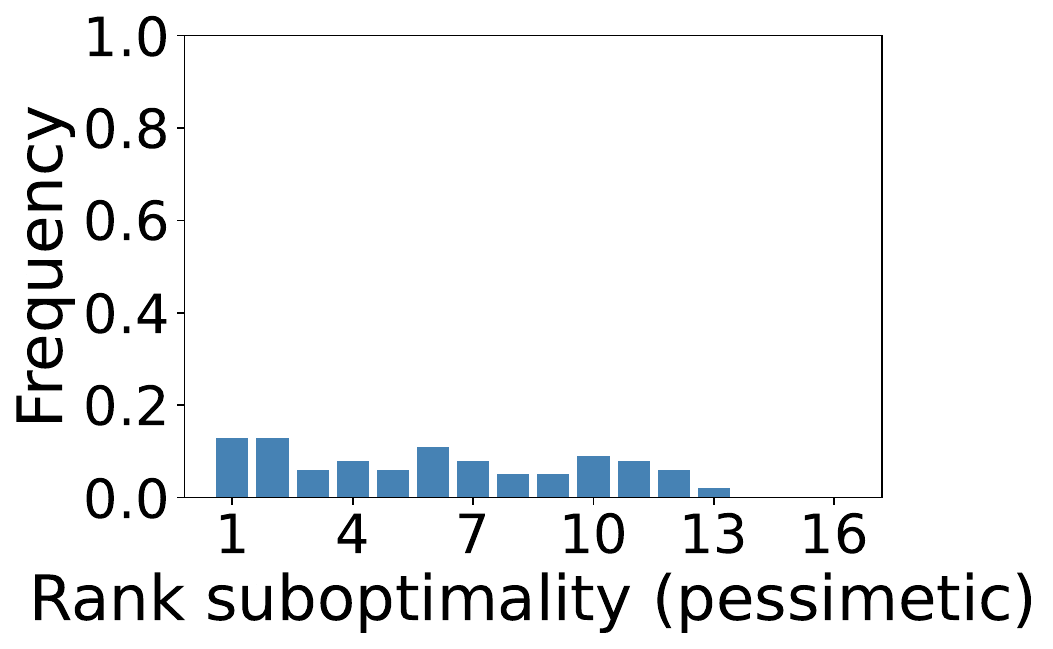}
\end{subfigure} \\
\begin{subfigure}[t]{0.32 \textwidth}
   \includegraphics[width=\linewidth]{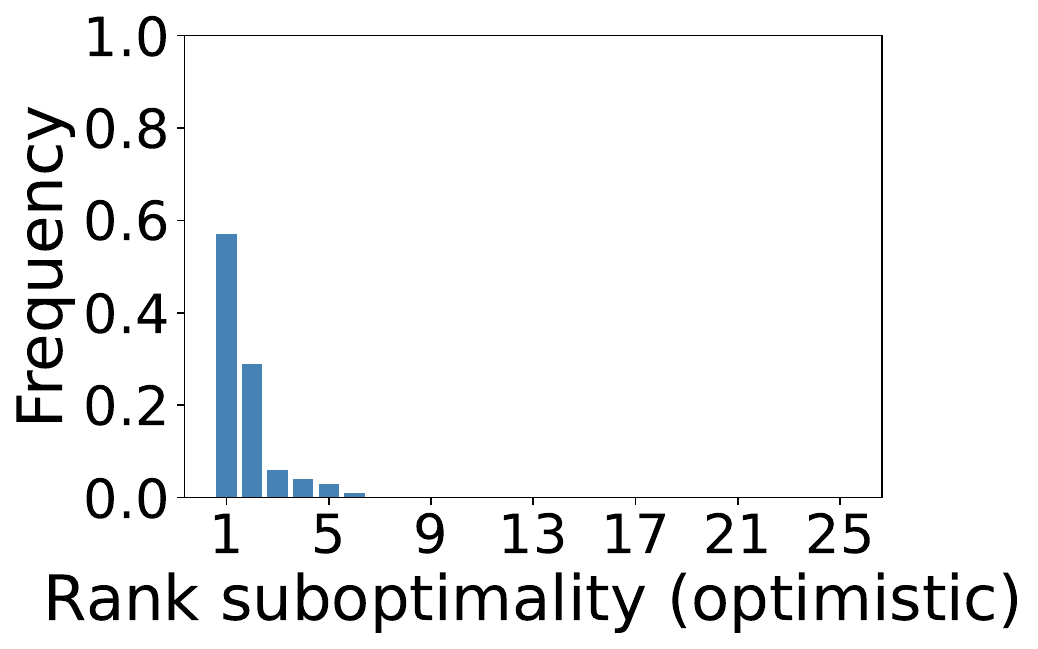}
\end{subfigure}
\begin{subfigure}[t]{0.32 \textwidth}
\includegraphics[width=\linewidth]{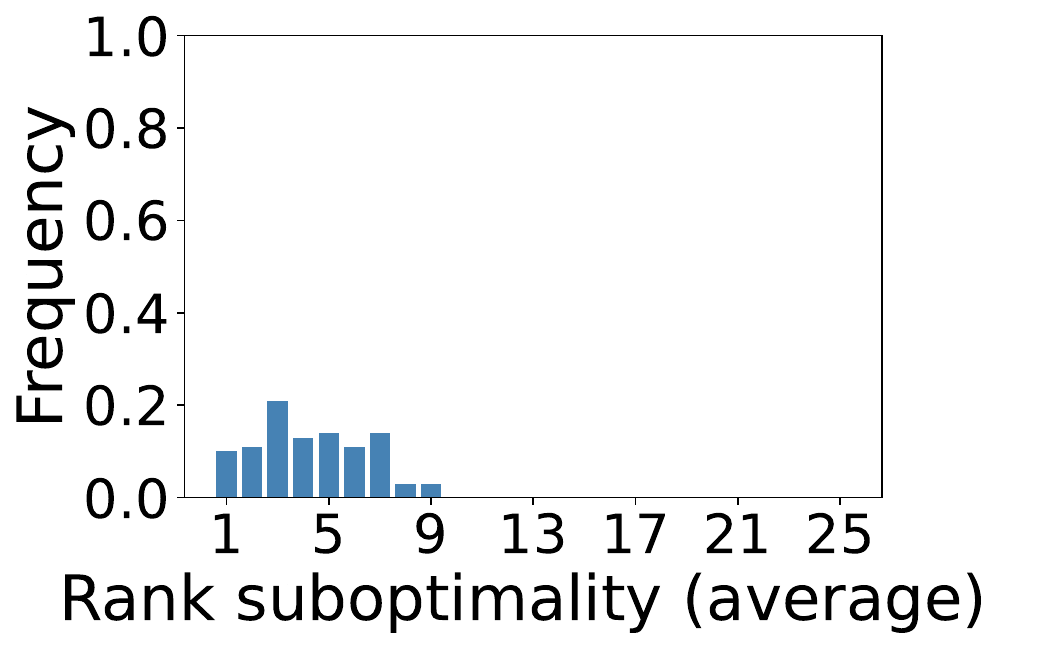}
\end{subfigure}
\begin{subfigure}[t]{0.32 \textwidth}
       \includegraphics[width=\linewidth]{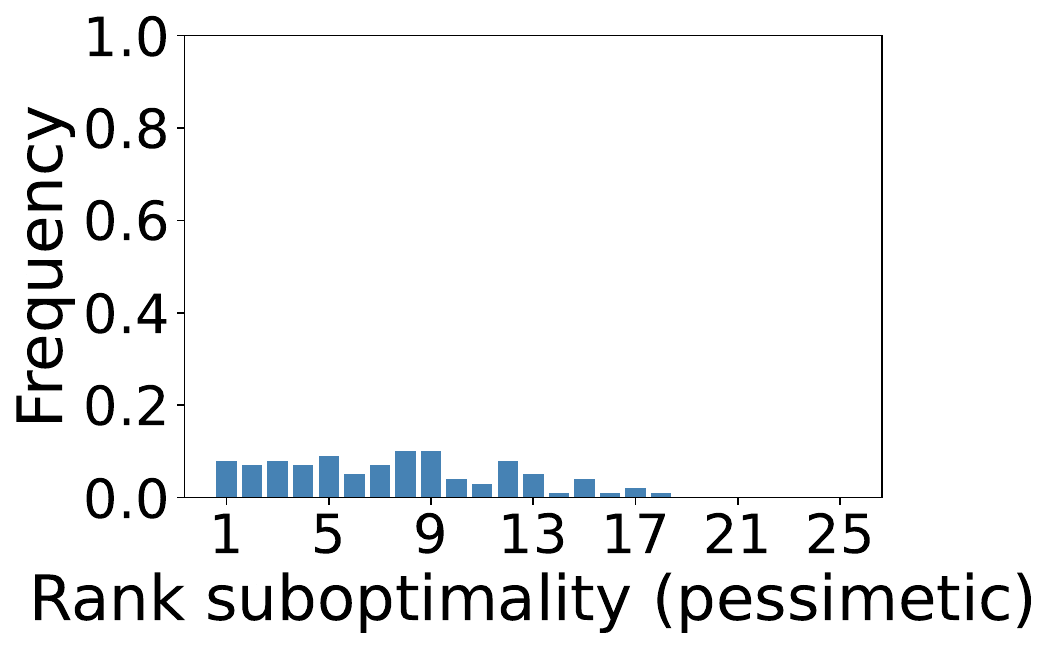}
\end{subfigure}
	\caption{Rank suboptimality. From top to bottom: RGS with 10 targets, schedules, and full suport, PSG with $m=4$, $r=2$, PLN with 5 layers each of width 5.}
	\label{fig:random_rank}

\end{figure}

\subsection{Quality of NSE computed}
We now investigate the quality of NSE that are computed with 2 defenders, under the SSAS assumption.
If there only defender $1$ existed (e.g., $V^2 = \{ 0 \})$), then defender $1$ simply chooses $v^1 = 0$ and $t$ to be its most desired target to be attacked. The existence of defender $2$ makes it such that both defenders must compromise to reach an NSE, with the target to be attacked worsened from defender $1$'s perspective. In essence, this is the ``cost of partnership''. 

We investigate this degradation in quality of attacked target (in ordinal terms) from the perspective of a single defender.
In each run, we compute all the \textit{efficient} NSE and their attacked target's rank in terms of the preference order $\succ_i$. This \textit{rank suboptimality} is a measure of the degradation of policy. 
Since there are multiple efficient NSE, we consider 3 cases: (i) the optimistic case where we tiebreak to benefit defender $1$, (ii) tiebreaking by averaging and (iii) the pessimistic case we tiebreak against defender $1$. For each of these settings, we run the experiment 100 times and report the frequency of every rank suboptimality. 
\begin{figure}[t]
\centering
\begin{subfigure}[t]{0.32 \textwidth}
    \includegraphics[width=\linewidth]{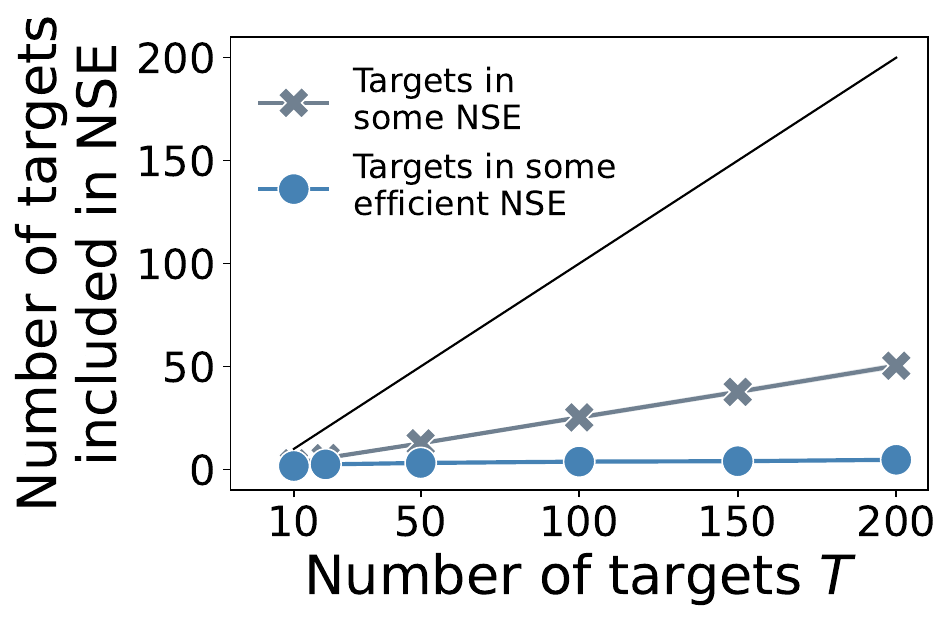}
\end{subfigure}
\begin{subfigure}[t]{0.32 \textwidth}
    \includegraphics[width=\linewidth]{figures/two_defender/NSE_t_eff_target.pdf}
\end{subfigure}
\begin{subfigure}[t]{0.32 \textwidth}
    \includegraphics[width=\linewidth]{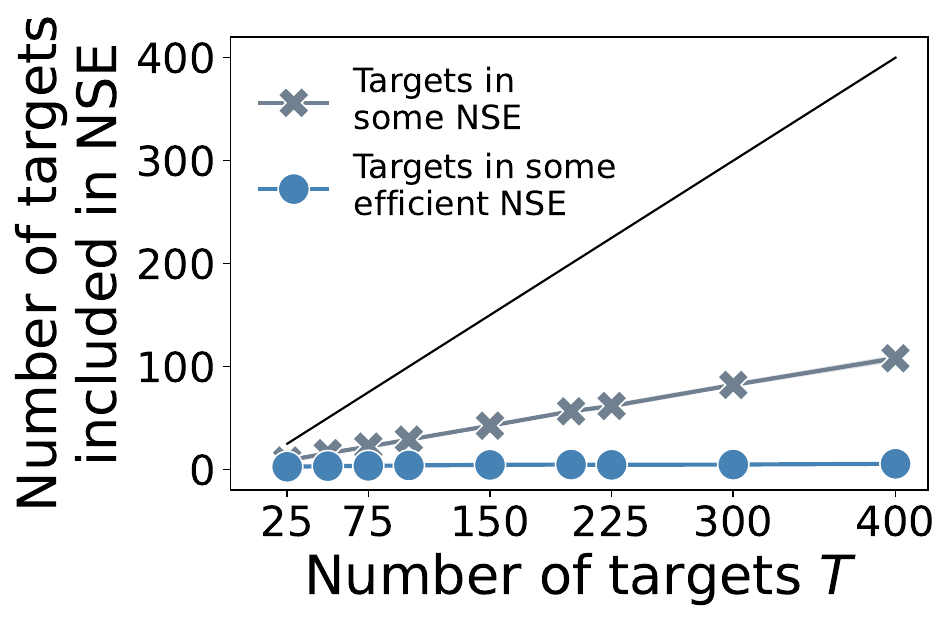}
\end{subfigure}\\

\begin{subfigure}[t]{0.32 \textwidth}
    \includegraphics[width=\linewidth]{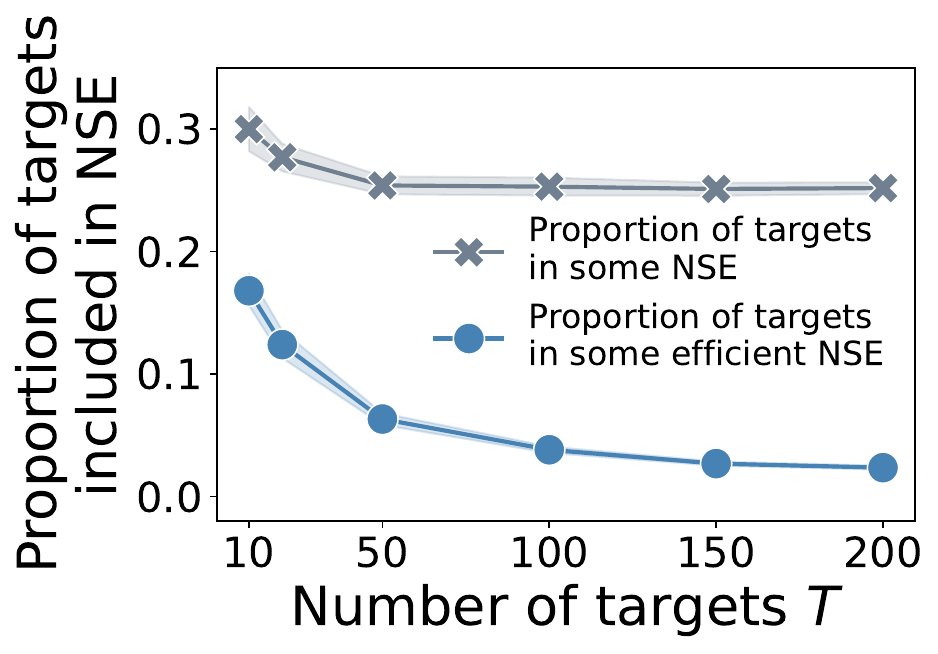}
\end{subfigure}
\begin{subfigure}[t]{0.32 \textwidth}
    \includegraphics[width=\linewidth]{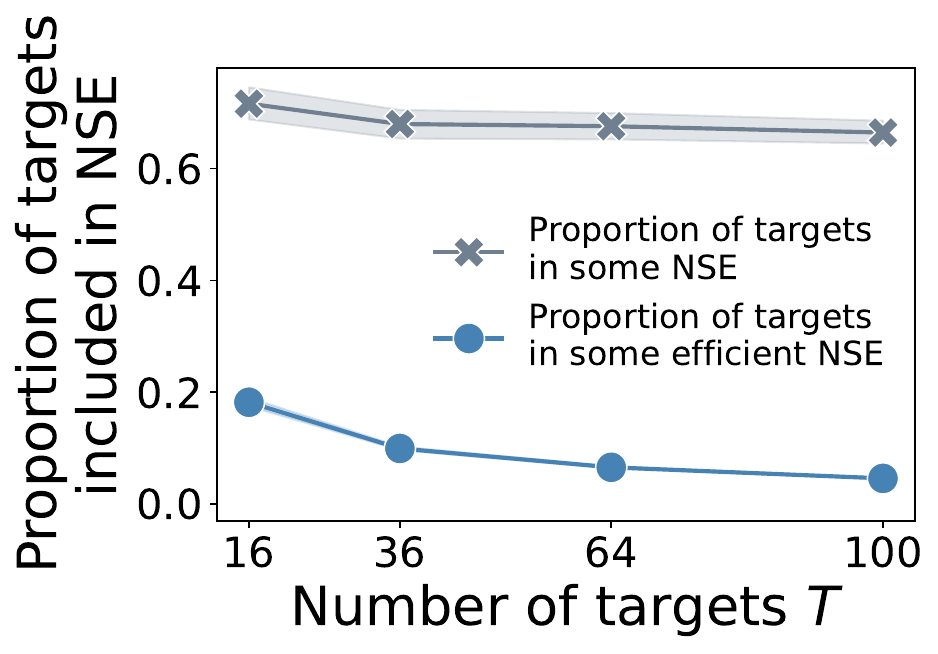}
\end{subfigure}
\begin{subfigure}[t]{0.32 \textwidth}
    \includegraphics[width=\linewidth]{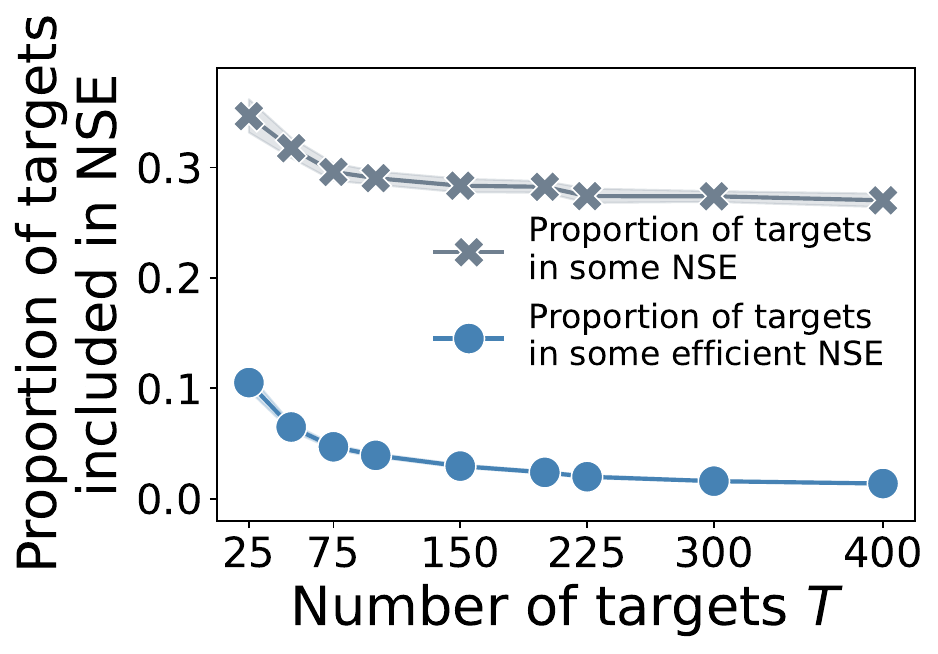}
\end{subfigure}\\
	\caption{Top (resp. bottom): \#targets (resp. ratio over $T$) that are efficient NSE as $T$ increases. Left to right: Results for RGS, PSG and PLN respectively.}
	\label{fig:types-of-targets}
\end{figure}

\subsection{Number of targets included in NSE}
Recall that each of the $T$ targets may be efficient, inefficient or not part of any NSE. 
We investigate for randomly generated $2$ defender games in RGS, PSG and PLN the number and proportion of these targets as $T$ varies. For RGS, we fix $S^i = 200$. Our results are reported in Figure~\ref{fig:types-of-targets}. We can see that in all our experiments, the number of efficient targets (or NSE) increase linearly with $T$, the \textit{proportion} of such targets decreases and tapers off at around $100$ targets.

\section{Conclusion}
In this paper, we explored the problem of multidefender security games in the presence of schedules in the restricted setting of coverage dependant utilities. 
We show that even in this restricted case, equilibrium may not exist under clearance constraints in contrast to prior work. We show that equilibrium is guaranteed under SSAS and present polynomial time solvers, as well as several extensions. 
Future work include removing the restriction on coverage dependant utilities as well as extensions to the non-additive or uncoordinated setting.
\section{Acknowledgements}
This research was sponsored by the U.S. Army Combat Capabilities Development Command Army Research Laboratory under CRA W911NF-13-2-0045. Co-author Fang is supported in part by NSF grant IIS-2046640 (CAREER) and Sloan Research Fellowship.

\section{Appendix}
\ifExtendedVersion
\subsection{Nonexistence of NSE in Example~\ref{example:counterexample}}
Let $\alpha=x^1_1$ and $1-\alpha=x^1_2$, i.e., the probability that defender $1$ plays schedule $1$ and $2$. Similarly, $\beta=x^2_1$ and $1-\beta=x^2_2$. Then, the total coverage $v^{total}$ (formatted as a $2\times 2$ matrix) is given by
\begin{align}
    v^{total} = \underbrace{
    \alpha \begin{bmatrix}
        1\shortminus \epsilon \ & 1 \\
        k\cdot \epsilon \ & 0
    \end{bmatrix} + 
    (1\shortminus\alpha) \begin{bmatrix}
        0 \ & k\cdot \epsilon \\
        1 \ & 1\shortminus \epsilon 
    \end{bmatrix}}_{v^1}
    + 
    \underbrace{
    \beta \begin{bmatrix}
        1 \ & 0 \\
        1 \shortminus \epsilon \ & k\cdot \epsilon 
    \end{bmatrix}
    + (1 \shortminus \beta) \begin{bmatrix}
        k\cdot \epsilon \ & 1 \shortminus \epsilon \\
        0 \ & 1
    \end{bmatrix}
    }_{v^2}
    \label{eq:alpha-beta}
\end{align}
Suppose $(\mathbf{v}, t)$ is a NSE. Without loss of generality, we assume $t=11$: as it turns out, symmetry in schedules implies the same argument will apply for other targets. A necessary condition for $(\mathbf{v}, 11)$ to be NSE is that it is $2$-IC, which in turn requires (but is not sufficient) that for deviation strategy $\widehat{v}^2 = s^2_1$ (i.e., player $2$ deviates by placing all their coverage using $s^2_1$) and potential post-deviation targets $12 \succ_2 11$ and $21 \succ_2 11$, we have
(i) either $12 \in B(\widehat{v}^{total})$ or $21 \in B(\widehat{v}^{total})$ and 
(ii) $11, 22 \not \in B(\widehat{v}^{total})$\footnote{Technically, only condition (ii) is required, since $B(\widehat{v}^{total})$ is nonempty.}
The same is also required of deviation strategy $\widehat{v}^2=s^2_2$.

More explicitly, a \textit{necessary} condition for $(\mathbf{v}, 11)$ to be $2$-IC is that $v^1$ is obtained from some $\alpha \in [0, 1]$ such that the following condition holds.
\begin{itemize}
    \item Defender $2$ deviating to $\widehat{v}^2$ by setting $\beta=1$ (i.e., playing $s^2_1$) does \textit{not} induce target $12$ (top right) to be attacked because either\footnote{Strictly speaking, for defender $2$ to induce target $12$ to be attacked, it needs ensure that target $21$ is not attacked as well. However, that instance would fall under the AND condition in the latter half and does not affect the correctness of our example.}
    \subitem 
    $\heartsuit$: $\widehat{v}^{total}(11) \leq \widehat{v}^{total}(12)$, implying target $11 \prec_2 12$ would be attacked. \footnote{The inequality is not strict due to tiebreaking rules after deviation.} 
    \\ OR
    \subitem $\varheart$: 
    $\widehat{v}^{total}(22) \leq \widehat{v}^{total}(12)$, implying target $22 \prec_2 12$ would be attacked.
    \item AND, Defender $2$ deviating to $\widehat{v}^2$ by setting $\beta=0$ (i.e., playing $s^2_2$) does \textit{not} induce target $21$ (bottom left) to be attacked because either
    \subitem $\diamondsuit$: 
    $\widehat{v}^{total}(22) \leq \widehat{v}^{total}(21)$, implying target $22 \prec_2 21$ would be attacked.\\
    OR
    \subitem 
    $\vardiamond$: $\widehat{v}^{total}(11) \leq \widehat{v}^{total}(21)$, implying target $11 \prec_2 21$ would be attacked.
\end{itemize}
It turns out this condition ($\heartsuit$ OR $\varheart$) AND ($\diamondsuit$ OR $\vardiamond$) is impossible to fulfill for suitable values of $\epsilon$ and $k$. To do this, we write down the expressions for $\heartsuit, \varheart, \diamondsuit, \vardiamond$. This is done by plugging in the values of $\beta$ (used for $\widehat{v}^{2}$) into \eqref{eq:alpha-beta} and expressing the relevant entry in the matrix used in the inequality. This yields the following inequalities:
\begin{align}
    \alpha \cdot (1-\epsilon) + 1 & \leq \alpha + (1-\alpha)\cdot k \cdot \epsilon \label{eq:heart} \tag{$\heartsuit$} \\
    (1-\alpha)\cdot(1-\epsilon) + k \cdot \epsilon &\leq \alpha + (1-\alpha)\cdot k \cdot \epsilon \label{eq:varheart} \tag{$\varheart$} \\
    (1-\alpha)\cdot (1-\epsilon) + 1 &\leq \alpha \cdot k \cdot \epsilon + (1-\alpha) \label{eq:vardiamond} \tag{$\diamondsuit$}\\
    \alpha \cdot (1-\epsilon) + k \cdot \epsilon &\leq \alpha \cdot k \cdot \epsilon + (1-\alpha) \label{eq:diamond} \tag{$\vardiamond$} 
\end{align}
These inequalities are \textit{linear} in $\alpha$. We now substitute $\epsilon = 10^{-3}$ and $k=100$, which after some algebra gives
\begin{align}
    \alpha &\leq -9.0909 \tag{$\heartsuit$} \\
    \alpha &\geq 0.5261 \tag{$\varheart$} \\
    \alpha &\geq 10.0909 \tag{$\diamondsuit$}\\
    \alpha &\leq 0.4739\tag{$\vardiamond$} 
\end{align}
Since $\heartsuit$ and $\diamondsuit$ require $\alpha$ to be out of the range $[0,1]$, condition ($\heartsuit$ OR $\varheart$) AND ($\diamondsuit$ OR $\vardiamond$) reduces to $\varheart$ AND $\vardiamond$. Clearly, this is impossible. Since $\varheart$ AND $\vardiamond$ is a necessary condition for $v^1$ to be part of any NSE given by $(\mathbf{v}, 11)$, we conclude that no NSE of such form exists.

It remains to repeat this process for the other possible equilibrium targets $t=22,t=12$ and $t=21$. In the former, we will once again show that no such $\alpha$ exists, while the latter 2 requires us to show that no such $\beta$ exists. The symmetry in our choice of schedules allows us to \textit{exactly} reuse the above arguments (same choice of $\epsilon$ and $k$) up to some rearrangements (e.g., swapping $\alpha$ with $\beta$). This concludes the proof that no NSE exists.

\subsection{NSE of Example~\ref{example:counterexample} under SSAS when $\epsilon=0$}
We now revisit Example~\ref{example:counterexample} when $\epsilon=0$. We claim that $(v^1, v^2, 11)$ where $v^1=(0, 0.5, 0.5, 0)$, $v_2 = (0, 0, 0, 1.0)$ is an NSE. Clearly $v^1 \in V^1$ and $v^2 \in V^2$ under SSAS. It is AIC since target $11$ has $0$ total coverage. It remains to verify that it is $1$ and $2$-IC. For defender $1$ to do better, it has to induce target $22$ to be attacked, by deviating to $\widehat{v}^1$ such that coverage on $11,12$ and $21$ is \textit{strictly} greater than $22$. This is impossible since $v^2(22)$ is already the highest at $1.0$. For defender $2$ to do better, it has to induce either $12$ or $21$ to get attacked. For that to happen, it has to ensure that $\widehat{v}^2(11)$ and $\widehat{v}^2(22)$ are \textit{both strictly greater} than $1/2$. 
This is once again impossible, since neither $s^2_1$ or $s^2_2$ covers both $11$ and $22$ simultaneously. It can place equal coverage over both targets, but this is insufficient to induce a change in attacker target since tiebreaks after deviation are against it. 

\subsection{NSE of Example~\ref{example:counterexample} under SSAS when $\epsilon=10^{-3}$, $k=100$}
The equilibrium is analogous to the case where $\epsilon=0$. Specifically, let $v^1=(0, 0.55, 0.55, 0)$ and $v^2=(0, 0, 0, 1.0)$, where coverages are for targets $11,12,21$ and $22$ respectively. We claim that $(\mathbf{v}, 11)=(v^1, v^2, 11)$ is a NSE. We first observe that $v^1 \in V^1$ and $v^2 \in V^2$ under SSAS.
Defender $1$ simply plays $s^1_1$ and $s^1_2$ with equal probability, but chooses to place no coverage on targets $11$ and $22$ (which it prefers to be attacked). On the other hand, defender $2$ plays $s^2_2$ exclusively and chooses to not place coverage on target $12$ (which defender $2$ prefers to be attacked). Lastly, we also verify that this is indeed AIC: target $11$ has $0$ coverage and thus certainly in $B(v^{total})$.

We can verify that $1$ and $2$-IC hold in the same manner as before. For defender $1$ to do any better that target $11$ being attacked, it has to induce an attack to target $22$. This can only be done by deviating to some $\widehat{v}^1 \in V^1$ that places \textit{strictly} more than $1$ coverage in \textit{each} target in $\{11,12,21\}$. However, this is impossible, since the sum of coverage (over targets) in either $s^1_1$ or $s^1_2$ does not exceed $3$, which in turn means that however defender $1$ splits coverage amongst targets $11,12,21$, at least one of them has coverage $\leq 1$. That is, $\widehat{v}^{total}(11) \leq 1$ or $\widehat{v}^{total}(12) \leq 1$ or $\widehat{v}^{total}(21) \leq 1$, while $\widehat{v}^{total} \geq 1$.
We have thus verified $1$-IC. To verify $2$-IC, note that defender $2$ has to find some $\widehat{v}^2 \in V^2$ such that \textit{both} target $11$ and $22$ have strictly higher coverage than $0.55$ (which was how much $v^1$ covered each of the other targets). However, the sum of coverage in targets $11$ and $22$ from either $s^2_1$ or $s^2_2$ is no greater than $1+k\cdot \epsilon = 1.1$. Again, this means that either $\widehat{v}^2(11) \leq 0.55$ or $\widehat{v}^2(22) \leq 0.55$. This implies that there is no way to induce an attack on $12$ or $21$ (and inducing an attack on $22$ is even less preferable for defender $2$). Thus $2$-IC holds, and this is indeed a NSE.
\fi
\subsection{Proof of Lemma~\ref{lem:step1}}
\begin{proof}
Since $v^1\in V^1$ and $\widetilde{v}^1(j) \leq v^1(j)$ for all $j \in [T]$, we have $\widetilde{v}^1\in V^1$ by Assumption~\ref{ass:SSAS}. 
We now show that $(\widetilde{v}^1, v^2, t)$ is AIC, $1$-IC, and $2$-IC.

\begin{enumerate}
    \item $(\widetilde{v}^1, v^2, t)$ is AIC. For all $j\in [T]$ and $j\neq t$, we have 
    \begin{align*}
        \widetilde{v}^1(j) + v^2(j) 
        \underbrace{= v^1(j) + v^2(j)}_{\text{definition of }\widetilde{v}^1}
        \ge \underbrace{v^1(t) + v^2(t)}_{(v^1, v^2, t) \text{ is AIC}}
        \ge \underbrace{\widetilde{v}^1(t) + v^2(t)}_{\text{by definition of }\widetilde{v}^1}. 
    \end{align*}

    \item $(\widetilde{v}^1, v^2, t)$ is 1-IC. If 
    not, there exists 
    $\widehat{v}^1\in \mathcal{V}^1, j\succ_1 t$ where $(\widehat{v}^1, v^2, j)$ is 1-WAIC, implying $(v^1, v^2, t)$ is not 1-IC and not an NSE. 

    \item $(\widetilde{v}^1, v^2, t)$ is 2-IC. 
    Suppose otherwise. 
    Then there exists $\widehat{v}^2\in \mathcal{V}^2, j \succ_2 t$ where $(\widetilde{v}^1, \widehat{v}^2, j)$ is 2-WAIC, implying that $\widetilde{v}^1(k) +\widehat{v}^2(k) > \widetilde{v}^1(j) +\widehat{v}^2(j)$ for $k \prec_2 j$, and $\widetilde{v}^1(k) +\widehat{v}^2(k) \ge \widetilde{v}^1(j) +\widehat{v}^2(j)$ for $k \succ_2 j$.
    Then for any $k \succ_2 j (\succ_2 t)$,  $\widetilde{v}^1(k) +\widehat{v}^2(k) \ge \widetilde{v}^1(j) +\widehat{v}^2(j)$ indicates that  $v^1(k) +\widehat{v}^2(k) \ge v^1(j) +\widehat{v}^2(j)$ by definition of $\widetilde{v}^1$ that $\widetilde{v}^1(k) = v^1(k), \widetilde{v}^1(j) = v^1(j)$. Besides, for any $k\prec_2 j$, 
    \begin{align*}
        v^1(k) + \widehat{v}^2(k) 
        \underbrace{\ge \widetilde{v}^1(k) + \widehat{v}^2(k)}_{\text{by definition of }\widetilde{v}^1}
        > \underbrace{\widetilde{v}^1(j) + \widehat{v}^2(j)}_{(\widetilde{v}^1, \widehat{v}^2, j)\text{ is 2-WAIC}}
        = \underbrace{v^1(j) + \widehat{v}^2(j)}_{\text{by definition of }\widetilde{v}^1}. 
    \end{align*}

    Since $(v^1, \widehat{v}^2, j)$ is 2-WAIC, $(v^1, v^2, t)$ is not 2-IC and not an NSE. 
\end{enumerate}
\end{proof}

\subsection{Proof of Lemma~\ref{lem:step2}}
\begin{proof}
$v^1\in V^1$ and $\widetilde{v}^1$ has no more coverage than $v^1$, so $\widetilde{v}^1\in V^1$ by the SSAS assumption. By the definition of NSE, it is sufficient to show that $(\widetilde{v}^1, v^2, t)$ is AIC, $1$-IC, and $2$-IC to prove the Lemma.
\begin{enumerate}
    \item $(\widetilde{v}^1, v^2, t)$ is AIC. $t\in B(\widetilde{v}^1, v^2)$ because $\widetilde{v}^1(t) + v^2(t) = 0$.
    \item $(\widetilde{v}^1, v^2, t)$ is 1-IC. If not, then there exists $\widehat{v}^1\in V^1, j\succ_1 t$ such that $(\widehat{v}^1, v^2, j)$ is 1-WAIC. Therefore, $(v^1, v^2, t)$ is also not 1-IC, contradicting the assumption that $(v^1, v^2, t)$ is an NSE.
    \item $(\widetilde{v}^1, v^2, t)$ is 2-IC. We prove this by contradiction. Suppose $(\widetilde{v}^1, v^2, t)$ is not 2-IC, then there exists $\widehat{v}^2\in V^2, j \succ_2 t$ where $(\widetilde{v}^1, \widehat{v}^2, j)$ is 2-WAIC. By definition of 2-WAIC, we have that $\widetilde{v}^1(k') +\widehat{v}^2(k') > \widetilde{v}^1(j) +\widehat{v}^2(j)$ for any target $k' \prec_2 j$, and $\widetilde{v}^1(k') +\widehat{v}^2(k') \ge \widetilde{v}^1(j) +\widehat{v}^2(j)$ for any target $k' \succ_2 j$.
    Consider $u^2\in V^2$ such that $u^2(k') = \widehat{v}^2(k')$ for any target $k' \preceq_2 t$ and $u^2(k') = 0$ for any target $k' \succ_2 t$. Notice that there always exist a unique target $k \in [T]$ such that $(v^1, u^2, k)$ is 2-WAIC. We claim that if $k$ is the target that $(v^1, u^2, k)$ is 2-WAIC, then $k \succ_2 t$. We prove this by excluding other targets $e \preceq_2 t$. Notice that there is a target $m \succ_2 t$ that $v^1(m) = \min_{{k'}\succ_2 t }v^1({k'})$. Consider a target $e \preceq_2 t$, then 
    \begin{align*}
        v^1(e) + u^2(e) 
        &\underbrace{\ge \widetilde{v}^1(e) + u^2(e)}_{\text{by definition of }\widetilde{v}^1}
        \underbrace{=\widetilde{v}^1(e) + \widehat{v}^2(e)}_{\text{by definition of }u^2}
        \underbrace{> \min_{{k'}\succ_2 t}\{\widetilde{v}^1({k'}) + \widehat{v}^2({k'})\}}_{(\widetilde{v}^1, \widehat{v}^2, j)\text{ is 2-WAIC}} \\
        & \underbrace{\ge \min_{{k'}\succ_2 t}\widetilde{v}^1({k'})}_{\widehat{v}^2({k'})\ge 0}
        \underbrace{ = v^1(m)}_{\text{by definition of }m}
        \underbrace{ = v^1(m) + u^2(m)}_{u^2(m) = 0}. 
    \end{align*}
    Thus, $e$ has more coverage than $m$, so $(v^1, u^2, e)$ is not 2-WAIC. Then $k \preceq_2 t$ does not hold. There exists a target $k \succ_2 t$ such that $(v^1, u^2, k)$ is 2-WAIC, so $(v^1, v^2, t)$ is not 2-IC, which contradicts that $(v^1, v^2, t)$ is an NSE. \hfill$\blacksquare$
\end{enumerate}
\end{proof}
\subsection{Proof of Theorem~\ref{thm:existence-2p}}
To better understand the structure of $\mathcal{H}$, we introduce \textit{partial} sets $\mathcal{H}^1_t\subset V^1, \mathcal{H}^2_t\subset V^2$ for $\mathcal{H}_t$ and show how they compose the set $\mathcal{H}_t$.
\begin{definition}
    $\mathcal{H}_t^1$ is the set of $v^1\in V^1$ such that there exists $h^1\ge 0$, $v^1(j) = h^1$ for any target $j \in \mathcal{T}_t^{\succ_2}$, $v^1(j) = 0$ for any target $j \in \mathcal{T}_t^{\preceq_2}$,and there does not exist $j\succ_2 t$ and $\widehat{v}^2\in V^2$, $(v^1, \widehat{v}^2, j)$ is 2-WAIC. Similar for $\mathcal{H}_t^2$.
    \label{def:H_partial}
\end{definition}

\begin{theorem}
    $\mathcal{H}_t = \mathcal{H}^1_t \times \mathcal{H}^2_t \times \{t\}$.
    \label{the:prod_thm}
\end{theorem}
\begin{proof}
    For any $(v^1, v^2, t)\in \mathcal{H}_t$, $v^1$ and $v^2$ are $t$-standard coverage. Since $\mathcal{H}_t$ only contains NSEs, $(v^1, v^2, t)$ is 1-IC and 2-IC. Thus, $v^1\in \mathcal{H}^1_t, v^2\in \mathcal{H}^2_t$. We now show that for any coverage $v^1\in \mathcal{H}_t^1$ and $v^2\in \mathcal{H}_t^2$, $(v^1, v^2, t)\in \mathcal{H}_t$. First, $(v^1, v^2, t)$ is an NSE. $(v^1, v^2, t)$ is AIC because $v^1(t) = v^2(t) = 0$. Also, $(v^1, v^2, t)$ is 1-IC and 2-IC by definiton of $\mathcal{H}^1_t$ and $\mathcal{H}^2_t$. Second, $v^1$ and $v^2$ are $t$-standard coverage. Thus, $(v^1, v^2, t) \in \mathcal{H}_t$. \hfill$\blacksquare$
\end{proof}

Theorem~\ref{the:prod_thm} decomposes the space of $\mathcal{H}_t$. Next we consider how to compute $\mathcal{H}^1_t$ and $\mathcal{H}^2_t$, which provides us an NSE in $\mathcal{H}_t$. We first consider the reduction of the set containing deviation strategies.
\begin{lemma}
    For a target $t$ and a coverage $v^1 \in \mathcal{H}_t^1$, if there is a $\widehat{v}^2\in V^2$ and $j \succ_2 t$ such that $(v^1, \widehat{v}^2, j)$ is 2-WAIC, we construct $u^2\in V^2$ such that $u^2(k) = 0$ for $k \in \mathcal{T}_t^{\succ_2}$ and $u^2(k) = \min_{k' \preceq_2 t} \widehat{v}^2(k')$ for $k \in \mathcal{T}_t^{\preceq_2}$, then there exists a target $m\succ_2 t$ such that $(v^1, u^2, m)$ is 2-WAIC.
    \label{lem:deviation_reduction}
\end{lemma}

\begin{proof}
    In a 2-WAIC strategy profile $(v^1, \widehat{v}^2, j)$ and $j\succ_2 t$, $v^1(k) + \widehat{v}^2(k) > v^1(j) + \widehat{v}^2(j)$ for any target $k\preceq_2 t$. Since $v^1\in \mathcal{H}_t^1$, $v^1(k) = 0$ for any target $k\preceq_2 t$. So, $\min_{k\preceq_2 t}\widehat{v}^2(k) > v^1(j) + \widehat{v}^2(j)$. We have $u^2(k) > v^1(j) + \widehat{v}^2(j)$ by definition of $u$, so any target $k\in \mathcal{T}_t^{\preceq_2}$ is not the attacked target. There always exists a target $m$ such that $(v^1, u^2, m)$ is 2-WAIC, and here $m \succ_2 t$. \hfill$\blacksquare$
\end{proof}


Lemma \ref{lem:deviation_reduction} reduces the deviation $\widehat{v}^2$ to $u^2$ with a canonical structure that $u^2(k)$ are equal for targets $k \in \mathcal{T}_t^{\preceq_2}$, but $u^2(k) = 0$ elsewhere. The computation of $u^2$ is related to the oracle $\textsc{MaximinCov}$. For convenience, we use $M^i(\mathcal{T})$ instead of $\textsc{MaximinCov}(\mathcal{T}, V^i)$ when the set of coverage is default to be $V^i$. Formally, $M^i(\mathcal{T}) = \max_{v^i \in V^i}\min_{j \in \mathcal{T}} v^i(j)$ for $\mathcal{T} \neq \emptyset$, and $M^i(\mathcal{T}) = + \infty$ for $\mathcal{T} = \emptyset$. With this notation, we give a sufficient and necessary condition for $\mathcal{H}_t^i \neq \emptyset$.

\begin{theorem}
    $\mathcal{H}_t^1 \neq \emptyset$ if and only if $M^1(\mathcal{T}_t^{\succ_2}) \ge M^2(\mathcal{T}_t^{\preceq_2})$. Similarly, $\mathcal{H}_t^2 \neq \emptyset$ if and only if $M^2(\mathcal{T}_t^{\succ_1}) \ge M^1(\mathcal{T}_t^{\preceq_1})$.
    \label{equivalent theorem}
\end{theorem}

\begin{proof}
    We prove the first claim for $\mathcal{H}_t^1\neq \emptyset$, and it is same for $\mathcal{H}_t^2\neq \emptyset$.
    \begin{enumerate}
        \item When $M^1(\mathcal{T}_t^{\succ_2}) \ge M^2(\mathcal{T}_t^{\preceq_2})$, we consider a coverage $v^1\in V^1$ of defender 1 such that $v^1(k) = M^1(\mathcal{T}_t^{\succ_2})$ for target $k \succ_2 t$, and $v^1(k) = 0$ elsewhere. For any coverage $\widehat{v}^2 \in V^2$, there is a target $m\in \mathcal{T}_t^{\preceq_2}$ such that $v^1(m) + \widehat{v}^2(m) \le M^2(\mathcal{T}_t^{\preceq_2}) \le v^1(k)$ for any target $k\in \mathcal{T}_t^{\succ_2}$. Any target $k\in \mathcal{T}_t^{\succ_2}$ has no less coverage than target $m$ given coverage $(v^1, \widehat{v}^2)$, so there does not exist target $j\succ_2 t$ such that $(v^1, \widehat{v}^2, j)$ is 2-WAIC. Thus, $v^1\in \mathcal{H}_t^1$ by definition.
        \item When $M^1(\mathcal{T}_t^{\succ_2}) < M^2(\mathcal{T}_t^{\preceq_2})$, we want to show that $\mathcal{H}_t^1 = \emptyset$. $M^1(\mathcal{T}_t^{\succ_2}) \neq +\infty$, so for any $v^1\in V^1$, there exists target $j\succ_2 t$ such that $v^1(j) \le M^1(\mathcal{T}_t^{\succ_2})$. Then consider a coverage $\widehat{v}^2\in V^2$ of defender 2 such that $\widehat{v}^2(k) = M^2(\mathcal{T}_t^{\preceq_2})$ for targets $k\in \mathcal{T}_t^{\preceq_2}$ and $\widehat{v}^2(k) = 0$ for other targets. In the coverage profile $(v^1, \widehat{v}^2)$, target $k\in \mathcal{T}_t^{\preceq_2}$ has strictly more coverage than target $j$. Notice that there always exists a target $j'$ such that $(v^1, \widehat{v}^2, j')$ is 2-WAIC. Thus, for any $v^1\in V^1$, there exists $\widehat{v}^2\in V^2$, such that there exists $j'\succ_2 t$ and $(v^1, \widehat{v}^2, j')$ is 2-WAIC. Therefore, $\mathcal{H}_t^1 = \emptyset$. \hfill$\blacksquare$
    \end{enumerate}
\end{proof}

\begin{corollary} 
    For any defender $i\in [2]$, there exists $t\in [T]$ such that $\mathcal{H}_t^i \neq \emptyset$.
    \label{corollary of equivalent}
\end{corollary}
\begin{proof}
    Let $i'\neq i$ be another defender.
    There exists a target $t$ such that for any target $j \in [T]$, $t \succeq_{i'} j$. Then $M^{i}(\mathcal{T}_t^{\succ_{i'}}) = M^{i}(\emptyset) = +\infty$. However, $M^{i'}(\mathcal{T}_k^{\preceq{i'}}) = M^{i'}([T])$ is finite. Therefore, $\mathcal{H}_t^i \neq \emptyset$. \hfill$\blacksquare$
\end{proof}

\begin{lemma}
    For any set of targets $\mathcal{T}' \subset \mathcal{T}$, $M_i(\mathcal{T}')\ge M_i(\mathcal{T})$ holds.
    \label{lem: subset}
\end{lemma}
\begin{proof}
    Let $v^{i*}$ be the coverage when $M_i(\mathcal{T})$ achieves the maximum. Let $v^i = v^{i *}$, then we have $\min_{t_j \in \mathcal{T}'} v^i(j) \ge M_{i}(\mathcal{T}).$ Therefore $M_i(\mathcal{T}') \ge M_i(\mathcal{T})$. \hfill$\blacksquare$
\end{proof}

\begin{lemma}
    For two defenders $i, i'$ and targets $t \succ_{i'} j$, if $\mathcal{H}_j^i \neq \emptyset$, then $\mathcal{H}_t^i \neq \emptyset$.
    \label{lem: one side}
\end{lemma}
\begin{proof}
    Since $t \succ_{i'} j$, we have $\mathcal{T}_t^{\succ_{i'}} \subset \mathcal{T}_j^{\succ_{i'}}$ and $\mathcal{T}_j^{\preceq_{i'}} \subset \mathcal{T}_t^{\preceq_{i'}}$. 
    By $\mathcal{H}_j^i \neq \emptyset$, we have $M^i(\mathcal{T}_j^{\succ_{i'}}) \ge M^{i'}(\mathcal{T}_j^{\preceq_{i'}})$. 
    Then using Lemma \ref{lem: subset}, we get $M^i(\mathcal{T}_t^{\succ_{i'}}) \ge M^i(\mathcal{T}_j^{\succ_{i'}}) \ge M^{i'}(\mathcal{T}_j^{\preceq_{i'}}) \ge M^{i'}(\mathcal{T}_t^{\preceq_{i'}})$. By Theorem \ref{equivalent theorem}, $\mathcal{H}_k^i \neq \emptyset$.  \hfill$\blacksquare$
\end{proof}

\begin{theorem}
$\mathcal{H}\neq \emptyset$.
\end{theorem}
\begin{proof} 
    Suppose $\mathcal{T}^1 = \{t\in [T] | \mathcal{H}_t^2 \neq \emptyset\}$ and $\mathcal{T}^2 = \{t\in [T] | \mathcal{H}_t^1 \neq \emptyset\}$. 
    By Corollary \ref{corollary of equivalent}, $\mathcal{T}^1, \mathcal{T}^2 \neq \emptyset$ and $\mathcal{T}^2 \neq \emptyset$. 
    There is a unique ${j_1} \in \mathcal{T}^1$ (${j_2} \in \mathcal{T}^2$) such that for any $t \in \mathcal{T}^1$, $t \succeq_1 {j_1}$ (for any $t \in \mathcal{T}^2$, $t \succeq_2 {j_2}$).
    It is sufficient to show that $\mathcal{T}^1 \cap \mathcal{T}^2 \neq \emptyset$. If we show this, then there is a target $j\in \mathcal{T}^1\cap \mathcal{T}^2$. We have $\mathcal{H}_j^1\neq \emptyset, \mathcal{H}_j^2\neq\emptyset$, and thus $\mathcal{H}\neq \emptyset$. Next we prove $\mathcal{T}^1 \cap \mathcal{T}^2 \neq \emptyset$ by contradiction.
    
    Suppose $\mathcal{T}^1 \cap \mathcal{T}^2 = \emptyset$. Since $\mathcal{T}^1, \mathcal{T}^2 \neq \emptyset$ and $\mathcal{T}^1 \cap \mathcal{T}^2 = \emptyset$, we have $\overline{\mathcal{T}^1} = [T]\setminus \mathcal{T}^1\neq \emptyset$, and $\overline{\mathcal{T}^2} = [T]\setminus \mathcal{T}^2\neq \emptyset$. 
    There is a unique ${k_1}\in \overline{\mathcal{T}^1}$ (${k_2}\in \overline{\mathcal{T}^2}$) such that for any $t\in \overline{\mathcal{T}^1}$, ${k_1} \succeq_1 t$ (for any $t\in \overline{\mathcal{T}^2}$, ${k_2} \succeq_2 t$). 
    By Lemma \ref{lem: one side}, for any target $t \succ_1 {j_1}$, $t \in \mathcal{T}^1$, and for any target $t \succ_2 {j_2}$, $t \in \mathcal{T}^2$. Therefore, ${j_1} \succ_1 {k_1}$ and ${j_2} \succ_2 {k_2}$. Then we have $\mathcal{T}_{k_1}^{\succ_1} = \mathcal{T}^1, \mathcal{T}_{k_1}^{\preceq_1} = \overline{\mathcal{T}^1}$, and $\mathcal{T}_{k_2}^{\succ_2} = \mathcal{T}^2, \mathcal{T}_{k_2}^{\preceq_2} = \overline{\mathcal{T}^2}$. By Theorem \ref{equivalent theorem} and $\mathcal{H}^2_{k_1} = \emptyset, \mathcal{H}^1_{k_2} = \emptyset$, we have $M^2(\mathcal{T}^1) < M^1(\overline{\mathcal{T}^1})$ and $M^1(\mathcal{T}^2) < M^2(\overline{\mathcal{T}^2})$. By $\mathcal{T}^1\cap \mathcal{T}^2 = \emptyset$, $\mathcal{T}^2 \subset \overline{\mathcal{T}^1}$ and $\mathcal{T}^1 \subset \overline{\mathcal{T}^2}$. By Lemma \ref{lem: subset}, $M^1(\overline{\mathcal{T}^1})< M^1(\mathcal{T}^2)$ and $M^2(\overline{\mathcal{T}^2})< M^2(\mathcal{T}^1)$. Combining this with previous inequality, we get $M^2(\mathcal{T}^1)< M^1(\mathcal{T}^2)$ and $M^1(\mathcal{T}^2) < M^2(\mathcal{T}^1)$, which are contradictory. 
    
    Therefore, $\mathcal{T}^1\cap \mathcal{T}^2 \neq \emptyset$. There is a target $j\in \mathcal{T}^1\cap \mathcal{T}^2$. By definition of $\mathcal{T}^1, \mathcal{T}^2$, we have $\mathcal{H}_j^1, \mathcal{H}_j^2\neq\emptyset$. Thus, $\mathcal{H}_j\neq \emptyset$ by Theorem \ref{the:prod_thm}. $\mathcal{H}_j\subset \mathcal{H}$, so $\mathcal{H}\neq \emptyset$. \hfill$\blacksquare$
\end{proof}

\subsection{Proof of Lemma~\ref{lem: no same cover}}
\begin{proof}
    By Theorem \ref{the:prod_thm} in the appendix, $\mathcal{H}_t = \mathcal{H}^1_t \times \mathcal{H}^2_t \times \{t\}$. $\mathcal{H}_t \neq \emptyset$, so $\mathcal{H}^1_t, \mathcal{H}^2_t \neq \emptyset$. Since $j \succ_1 t, j \succ_2 t$, further by Lemma \ref{lem: one side} in the appendix, we have $\mathcal{H}^1_j, \mathcal{H}^2_j \neq \emptyset$. Use Theorem \ref{the:prod_thm} again, $\mathcal{H}_j = \mathcal{H}^1_j \times \mathcal{H}^2_j \times \{j\}$. Therefore, $\mathcal{H}_j \neq \emptyset$. \hfill$\blacksquare$
\end{proof}

\subsection{Proof of Theorem~\ref{thm:multiNSEexist} where $F(t)\neq F(t')$}
\begin{proof}


Here, $\Argmin_{t} F(t)$ may not be a singleton, and we
handle the edge case of selecting a specific $k^*$ from $\Argmin_{t} F(t)$.
Consider the set $\mathcal{F} = \Argmin_{t} F(t)$ and let $\underline{F} = \min_{t'\in [T]} F(t')$. For any $t, j\in \mathcal{F}, t\neq j$, we say $t \triangleleft j$ if for any defender $i$, $\mathcal{M}_{i, \pi_i(t)} = \underline{F} \implies t \succ_i j$. 
We show that $\triangleleft$ is transitive. Suppose $t\triangleleft j$ and $j\triangleleft j'$. 
Hence, $t \succ_i j$ for any defender $i$ where $\mathcal{M}_{i, \pi_i(t)} = \underline{F}$. 
On one hand, $\mathcal{M}_{i, \pi_i(j)}\ge \mathcal{M}_{i, \pi_i(t)} \ge \underline{F}$ since $t\succ_i j$. On the other hand, $j\in \mathcal{F}$, so $\mathcal{M}_{i, \pi_i(j)} \le \underline{F}$. Therefore, $\mathcal{M}_{i, \pi_i(j)} = \underline{F}$. 
Using a similar argument on $j\triangleleft j'$ 
gives us 
$\mathcal{M}_{i, \pi_i(t)} = \mathcal{M}_{i, \pi_i(j')} = \underline{F}$ and $t\succ_i j\succ_i j'$. 
So, $t\triangleleft j'$ and $\triangleleft$ is transitive.

We claim there exists a target $k^*\in \mathcal{F}$ such that $\nexists t\in \mathcal{F}$, $t \triangleleft k^*$. 
If not, then for any $k_q\in \mathcal{F}$, there exists $k_{q+1}\in \mathcal{F}$ where $k_{q+1} \triangleleft k_q$. 
Thus, we can construct an infinite sequence $\cdots \triangleleft k_2 \triangleleft k_1$. Since $\triangleleft$ is transitive, each $k_q$ must be distinct, which is not possible with a finite number of targets. Now we select a target $k^*$.


Once $k^*$ is selected, we set $v^i(k^*) = 0$ for all defenders $i$. For any $t\neq k^*$, we have $F(t) \ge F(k^*)$. For a fixed $t$, there are two cases: 
(i) if $F(t) > F(k^*)$, for any defender $i$ such that $\mathcal{M}_{i, \pi_i(t)} = F(t)$, we have $k^* \succ_i t$, since $\mathcal{M}_{i, \pi_i(t)} > F(k^*) \ge \mathcal{M}_{i, \pi_i(k^*)}$. 
(ii) If $F(t) = F(k^*)$, there is at least one defender $i$ satisfying $\mathcal{M}_{i, \pi_i(t)} = F(t)$ such that $k^* \succ_i t$, otherwise $t\triangleleft k^*$, contradicting the choice of $k^*$. Thus, for any $t\neq k^*$, there always exists a defender $i$ where $\mathcal{M}_{i, \pi_i(t)} = F(t)$ and $k^* \succ_i t$. 
Our construction sets $v^i(t) = F(k^*)$ and $v^{i'}(t) = 0$ for $i'\neq i$. 
Showing $\mathbf{v}$ is feasible and $(\mathbf{v}, k^*)$ is an NSE is same with the main proof of Theorem~\ref{thm:multiNSEexist}.\hfill$\blacksquare$

\end{proof}

%
%
%
\bibliographystyle{splncs04}
\bibliography{main}
%




\end{document}